\documentclass[11pt]{article}

\usepackage{makeidx}
\usepackage{color}
\usepackage{graphicx}
\usepackage{hyperref}
\usepackage{caption}
\usepackage{subcaption}
\usepackage[margin=1in]{geometry}
\usepackage{tikz}
\usepackage{tikzscale}
\usepackage{pgfplots}
\usepackage{pgfplotstable}
\usetikzlibrary {positioning}
\usepackage[makeroom]{cancel}
\usepackage[toc,page]{appendix}
\usepackage [
n,
advantage,
operators,
sets,
adversary,
landau,
probability,
notions,
logic,
ff,
mm,
primitives,
events,
complexity,
asymptotics,
keys
] {cryptocode}
\usepackage{qcircuit}
\usepackage[affil-it]{authblk}
\usepackage{mathtools}
\usepackage{amssymb, amsthm, amsmath}
\usepackage{braket}
\usepackage{MyMnsymbol}
\usepackage{algorithm, algorithmicx, algpseudocode}
\floatname{algorithm}{Protocol}
\usepackage{physics}
\usepackage{float}
\usepackage{enumitem}
\usepackage[textsize=small]{todonotes}
\usepackage[makeroom]{cancel}
\def\EQ#1{\begin{eqnarray}#1\end{eqnarray}}

\usepackage{breakcites}
\allowdisplaybreaks % Breaks big equations

\newcommand*{\cA}{\mathcal{A}}

\newcommand*{\cD}{\mathcal{D}}

\newcommand*{\cF}{\mathcal{F}}
\newcommand*{\cG}{\mathcal{G}}

\newcommand*{\cP}{\mathcal{P}}

\newcommand*{\Z}{\mathbb{Z}}

\newcommand*{\R}{\mathbb{R}}

\newcommand{\LWE}{\textsc{LWE}}
\newcommand{\GapSVP}{\textsc{GapSVP}}
\newcommand{\SIVP}{\textsc{SIVP}}

\newcommand\tab[1][1cm]{\hspace*{#1}}
\newcommand*{\eps}{\varepsilon}

\newtheorem{theorem}{Theorem}[section]

\newtheorem{corollary}[theorem]{Corollary}

\newtheorem{lemma}[theorem]{Lemma}

\newtheorem{definition}{Definition}[section]

\newtheorem{remark}{Remark}[section]

\newtheorem*{remark*}{Remark}

\numberwithin{algorithm}{section}

\title{On the possibility of classical client blind quantum computing} 

\author[1]{Alexandru Cojocaru}
\author[2,3]{L\'{e}o Colisson}
\author[1,2]{Elham Kashefi}
\author[1]{Petros Wallden}
\affil[1]{School of Informatics, University of Edinburgh,}
\affil[ ]{10 Crichton Street, Edinburgh EH8 9AB, UK}
\affil[2]{D\'{e}partement Informatique et R\'{e}seaux, LIP6, Sorbonne Universit\'{e},}
\affil[ ]{4 Place Jussieu 75252 Paris CEDEX 05, France}
\affil[3]{Ecole Normale Sup\'{e}rieure Paris-Saclay,}
\affil[ ]{61 Avenue du Pr\'{e}sident Wilson, 94230 Cachan, France}

\date{Last update: \today\footnote{Original article: ``Delegated Pseudo-Secret Random Qubit Generator'', February 27, 2018}}

\begin{document}

\maketitle

\begin{abstract}
We define the functionality of delegated pseudo-secret random qubit generator (PSRQG), where a classical client can instruct the preparation of a sequence of random qubits at some distant party. Their classical description is (computationally) unknown to any other party (including the distant party preparing them) but known to the client. We emphasize the unique feature that no quantum communication is required to implement PSRQG. This enables classical clients to perform a class of quantum communication protocols with only a public classical channel with a quantum server. A key such example is the delegated universal blind quantum computing. Using our functionality one could achieve 
a purely classical-client computational secure verifiable delegated universal quantum computing (also referred to as verifiable blind quantum computation). We give a concrete protocol (QFactory) implementing PSRQG, using the Learning-With-Errors problem to construct a trapdoor one-way function with certain desired properties (quantum-safe, two-regular, collision-resistant). We then prove the security in the Quantum-Honest-But-Curious setting and briefly discuss the extension to the malicious case.
\end{abstract}

\section{Introduction and Related Works}

The recent interest in quantum technologies has brought forward a vision of \emph{quantum internet} \cite{elkouss2017} that could implement a collection of known protocols for enhanced security or communication complexity (see a recent review in \cite{BC2016}). On the other hand the rapid development of quantum hardware has increased the computational capacity of quantum servers that could be linked in such a communicating network. This raised the necessity/importance of privacy preserving functionalities such as the research developed around quantum computing on encrypted data (see a recent review in \cite{fitzsimons2017}).

However, there exist some challenges in adapting widely the above vision: A reliable long-distance quantum communication network connecting all the interested parties might be very costly. Moreover, currently, some of the most promising quantum computation devices (e.g. superconducting such as the devices developed by IBM, Google, etc) do not yet offer the possibility of ``networked'' architecture, i.e. cannot receive and send quantum states.

For this reason, there has been extensive research focusing on the practicality aspect of quantum delegated computation protocols (and related functionalities). One direction is to reduce the required communications by exploiting classical fully-homomorphic-encryption schemes \cite{broadbent2015quantum,dulek2016quantum,alagic2017quantum}, or by defining their direct quantum analogues \cite{liang2015quantum,ouyang2015quantum,tan2016quantum,lai2017statistically}. Different encodings, on the client side, could also reduce the communication \cite{mantri2013optimal,GMMR2013}. However, in all these approaches the client still requires some quantum capabilities. While no-go results indicate restrictions on which of the above properties are jointly achievable for classical clients  \cite{armknecht2014general,yu2014limitations,ACGK2017,newman2017limitations}, completing this picture remains an open problem. Another direction is to consider fully-classical client protocols, compatible with the no-go results, that can therefore achieve more restricted levels of security. The first such procedure achieving statistical security (but not for universal computations) was proposed in \cite{mantri2017flow}.  Focusing on post-quantum computational security a universal blind delegated protocol was proposed in \cite{mahadev2017} and a verifiable one in \cite{mahadev2018}. 

Our own independent work presented here, is also based on post-quantum computational security, appeared (in preprint \cite{CCKW18}) in between the above mentioned two works, taking a different approach, more natural to measurement-based quantum computing protocols. The approach we take is modular. We replace the need for (a particular) quantum communication channel with a computationally (but post-quantum) secure generation of secret and random qubits. This can be used by classical clients to achieve blind quantum computing and a number of other applications.

\subsection{Our Contributions}

\begin{enumerate}
\item We define a classical client/quantum server delegated ideal functionality of pseudo-secret random qubit generator (PSRQG), in \autoref{Sec:Ideal}. PSRQG can replace the need for quantum channel between parties in certain quantum communication protocols with trade-off that the protocols become computationally secure (against \emph{quantum} adversaries).
\item We give a basic protocol (QFactory) that achieves this functionality, given a trapdoor one-way function that is quantum-safe, two-regular and collision resistant resistant in \autoref{Sec:Protocol} and prove its correctness.
\item We prove the security of the QFactory against Quantum-Honest-But-Curious server or against any malicious third party by proving that the classical description of the generated qubits is a hard-core function (following a reduction similar that of the Goldreich-Levin Theorem) in \autoref{Sec:Privacy}.
\item While our previous results do not depend on the specific function used, the existence of such specific functions (with all desired properties) makes the PSRQG a practical primitive that can be employed as described in this paper. In \autoref{Sec:Functions}, we first give methods for obtaining two-regular trapdoor one-way functions with extra properties (collision resistant or second preimage resistant) assuming the existence of simpler trapdoor one-way functions (permutation trapdoor or homomorphic trapdoor functions). 
We use reductions to prove that the resulting functions maintain all the properties required. Furthermore, we give in \autoref{Subsec:actual_trapdoor} an explicit family of functions that respect all the required properties based on the security of the Learning-With-Errors problem as well as a possible instantiation of the parameters. Thus, this function is also quantum-safe, and thus directly applicable for our setting. Note, that other functions may also be used, such as the one in \cite{mah2018} or functions based on the Niederreither cryptosystem and the construction in \cite{freeman2010}.
  
\end{enumerate}

\subsection{Applications}\label{Sec:applications} 

The PSRQG functionality, viewed as a resource, has a wide range of applications. Here we give a general overview of the applications, while for details on \emph{how} to use the \emph{exact output} of the PSRQG obtained in this paper in specific protocols we refer the reader to \autoref{app:applications}. PSRQG enables fully-classical parties to participate in many quantum protocols using only public classical channels and a single (potentially malicious) quantum server.

\textbf{The first type of applications} concerns a large class of delegated quantum computation protocols, including blind quantum computation and \emph{verifiable} blind quantum computation. These protocols are of great importance, enabling information-theoretically secure (and verifiable) access to a quantum cloud. However, the requirement for quantum communication limits their domain of applicability. This limitation is removed by replacing the off-line preparation stage with our QFactory protocol. Concretely, we can use QFactory to implement the blind quantum computation protocol of \cite{bfk}, as well as the \emph{verifiable} blind quantum computation protocols (e.g. those in \cite{fk,Broadbent2015,FKD2017}), in order to achieve classical-client secure and verifiable access to a quantum cloud.

In all these cases, the cost of using PSRQG is that the security becomes post-quantum computational (from information-theoretic). However, the possibility of information-theoretically secure classical client blind quantum computation seems highly unlikely due to strong complexity-theoretic arguments given in \cite{ACGK2017} and therefore this is the best we could hope for.

\textbf{The second type of applications} involves the family of protocols for which their quantum communication consists of random single qubits as the ones provided by QFactory, such as: quantum-key-distribution \cite{BB84}, 
quantum money \cite{BOV2018}, 
quantum coin-flipping \cite{PCDK2011}, 
quantum signatures \cite{WDKA2015}, 
two-party quantum computation \cite{KW2017,KMW2017}, 
multiparty quantum computation \cite{KP16}, 
etc.

Finally, we note that in order to use PSRQG as a subroutine in a larger protocol, we need to address the issue of composition and formulate the functionality in the universal composability framework \cite{unruh2010uni}.  This could be done as in \cite{DK2016} (where \emph{quantum} communication was required, using a quantum version of SRQG), but the full details are outside of the scope of this paper. 

\subsection{Overview of the Protocol and Proof}\label{Sec:Overview}

The general idea is that a classical client gives instructions to a quantum server to perform certain actions (quantum computation). Those actions lead to the server having as output a single qubit, which is randomly chosen from within a set of possible states of the form $\ket{0}+e^{ir\pi/4}\ket{1}$, where $r\in\{0,\cdots,7\}$. The randomness of the output qubit is due to the (fundamental) randomness of quantum measurements that are part of the instructions that the client gives. Moreover, the server cannot guess the value of $r$ any better than if he had just received that state directly from the client  (up to negligible probability). This is possible because the instructed quantum computation is generically a computation that is hard to (i) classically simulate and (ii) to reproduce quantumly because it is unlikely (exponentially in the number of measurements) that by running the same instructions the server obtains the exact same measurement outcomes twice. On the other hand, we wish the client to \emph{know} the classical description and thus the value of $r$. To achieve this task, the instructions/quantum computation the client uses are based on a family of trapdoor one-way functions with certain extra properties\footnote{The functions should also be two-regular (each image has exactly two preimages), quantum safe (secure against quantum attackers) and collision resistant (hard to find two inputs with the same image).}. Such functions are hard to invert (e.g. for the server) unless someone (the client in our case) has some extra ``trapdoor'' information $t_k$. This extra information makes the quantum computation easy to classically reproduce for the client, which can recover the value $r$, while it is still hard to classically reproduce for the server. Sending random qubits of the above type, is exactly what is required from the client in most of the protocols and applications given earlier, while with simple modifications our protocol could achieve other similar sets of states. \\ \\ \\
Our QFactory protocol can heuristically be described in the next steps:

\noindent \textbf{Preparation.} The client randomly selects a function $f_k$, from a family of trapdoor one-way, quantum-safe, two-regular and collision resistant functions. The choice of $f_k$ is public (server knows), but the trapdoor information $t_k$ needed to invert the function is known only to the client.

\noindent \textbf{Stage 1: Preimages Superposition.} The client instructs the server 
(i) to apply Hadamard(s) on the control register, (ii) to apply $U_{f_k}$ on the target register i.e. to obtain $\sum_x \ket{x}\otimes\ket{f_k(x)}$ and (iii) to measure the target register in the computational basis, in order to obtain a value $y$. This collapses his state to the state $(\ket{x}+\ket{x'})\otimes\ket{y}$, where $x,x'$ are the unique two preimages of $y$.

\emph{Remarks.} First we note that each image $y$ appears with same probability (therefore, obtaining twice the same $y$ happens with negligible probability). We now consider the first register $\ket{x}+\ket{x'}=\ket{x_1\cdots x_n}+\ket{x'_1\cdots x'_n}$, where the subscripts denote the different bits of the corresponding preimages $x$ and $x'$. We rewrite this: 
\begin{align*}
\big(\otimes_{i\in \bar{G}}\ket{x_{i}}\big)\otimes\Big(\prod_{j\in G}X^{x_{j}}\Big)\big(\ket{0\cdots 0}_{G}+\ket{1\cdots 1}_{G}\big)
\end{align*} 
where $\bar{G}$ is the set of bits positions where $x,x'$ are identical, $G$ is the set of bits positions where the preimages differ, while we have suitably changed the order of writing the qubits. It is now evident that the state at the end of Stage 1 is a tensor
product of isolated $\ket{0}$ and $\ket{1}$ states, and a Greenberger-Horne-Zeilinger (GHZ)
state with random $X$'s applied. The crucial observation is that the connectivity (which qubit belongs to the GHZ and which doesn't) depends on the XOR of the two preimages $x\oplus x'$ and is computationally impossible to determine, with non-negligible advantage, without the trapdoor information $t_k$.

\noindent \textbf{Stage 2: Squeezing.} The client instructs the server to measure each qubit $i$ (except the output) in a random basis $\{\ket{0}\pm e^{i\alpha_i\pi/4}\ket{1}\}$ and return back the measurement outcome $b_i$. The output qubit is of the form $\ket{+_\theta}=\ket{0}+e^{i\theta}\ket{1}$, where (see \cite{CCKW18}):
\begin{equation}
\label{eq:theta0} \theta=\frac{\pi}{4}(-1)^{x_n}\sum_{i=1}^{n-1}(x_i-x_i')(4b_i+\alpha_i)\bmod 8
\end{equation}
Intuitively, measuring qubits that are not connected has no effect to the output, while measuring qubits within the GHZ part, rotates the phase of the output qubit (by a $(-(1)^{x_i}\alpha_i+4b_i)\pi/4$ angle).

\noindent \textbf{Security.} The protocol is secure, if we can prove that the server (or other third parties) cannot guess (obtain noticeable advantage in guessing) the classical description of the state, i.e. the value of $\theta$. We consider a quantum-honest-but-curious server (see formal definition below) which means that he essentially follows the protocol and the security reduces in proving that the server cannot use his classical information to obtain \emph{any} advantage in guessing the classical description of the (honest) quantum output.

The server does not know the two preimages $x,x'$ and needs to guess $\theta$ from the value of the image $y$. A similar (simpler) result that we use is the Goldreich-Levin theorem \cite{GoldreichLevin}, that (informally) states that the inner product of the preimage of a one-way function with a random vector, taken modulo 2, is indistinguishable from a random bit. Our case is similar since Eq. (\ref{eq:theta0}) has the form of an inner product of the XOR of two preimages with a random vector taken modulo 8. We prove that if a computationally bounded server could obtain non-trivial advantage in guessing $\theta$, then he could also break the property of ``second preimage resistance'' which we requested for our function $f_k$.

\noindent \textbf{The function.} Our protocol relies on using functions that have a number of properties (one-way, trapdoor, two-regular, collision resistant (see \autoref{rmk:collision_resistance})), quantum safe). \emph{Any} function satisfying those conditions is suitable for our protocol. While in first thought some of these appear hard to satisfy jointly (e.g. two-regularity and collision resistance), we give two constructions that achieve those properties from simpler functions: one from injective, homomorphic trapdoor one-way function and one from bijective trapdoor one-way function. Both constructions define a new function that has domain extended by one bit, and the value of that bit ``decides'' whether one uses the initial basic function or not.

We then use a (slight) modification of the first construction and the trapdoor one-way function based on the Learning-with-Errors of \cite{MP2012} with suitable choice of parameters, and obtain a function that has all the desired properties. In a nutshell, the idea is to use the construction of \cite{MP2012}, to create an injective function $g(s,e)$ hard to invert without the secret trapdoor, and then to sample from a Gaussian distribution a small error term $e_0 \in Z_q^m$ as well as a (uniform) random $s_0 \in \Z_q^n$. According to \cite{MP2012}, it should be impossible to recover efficiently $s_0$ and $e_0$ from $b_0 := g(s_0,e_0)$. Then, to create the function $f(s,e,c)$, we define $f(s,e,0) = g(s,e)$ and $f(s,e,1) = g(s,e) + b_0$, and we require $e$ to have infinity norm smaller than a parameter $\mu$. Because the function is ``nearly homomorphic'', it appears that $f(s,e,1) = f(s+s_0, e+e_0, 0)$, so this function has intuitively two preimages. However, $e+e_0$ may not be small enough to stay in the input domain, so it may be possible to have only one preimage for some $y$. What we show is that if we sample $e_0$ ``small enough'' (at least as small as $O(\mu/m)$), then the probability to have two preimages is at least constant. Moreover, we prove that this modification does not break the security of $g$, and leads to a function $f$ that is both one-way and collision resistant under the $LWE$ assumption, which reduces to $\SIVP_\gamma$, with $\gamma = \poly[n]$.

\section{Preliminaries}\label{Sec:Prelim}

\subsection{Classical Definitions}

We are considering protocols secure against quantum
adversaries, so we assume that all the properties of our functions hold for a general Quantum Polynomial Time (QPT) adversary, rather than the usual Probabilistic Polynomial Time (PPT) one. We will denote $D$ the domain of the functions, while $D(n)$ is the subset of strings of length $n$.

\begin{definition}[Quantum-Safe (informal)]
A protocol/function is quantum-safe (also known as post-quantum secure), if all its properties remain valid when the adversaries are QPT (instead of PPT).
\end{definition}
\noindent The following definitions are for PPT adversaries, however in this paper we will generally use quantum-safe versions of those definitions and thus security is guaranteed against QPT adversaries.

\begin{definition}[One-way]\label{def:one_way_function}
A family of functions $\{f_k : D \rightarrow R\}_{k \in K}$ is \textbf{one-way} if:
\begin{itemize}
\item There exists a PPT algorithm that can compute $f_k(x)$ for any index function $k$, outcome of the PPT parameter-generation algorithm \text{Gen} and any input $x \in D$;
\item Any PPT algorithm $\cA$ can invert $f_k$ with at most negligible probability over the choice of $k$: \\
  $ \underset{\substack{k \leftarrow Gen(1^n) \\  x \leftarrow D \\ rc \leftarrow \{0, 1\}^{*}}} \Pr [f(\mathcal{A}(k, f_k(x)) = f(x)] \leq \negl$\\
  where $rc$ represents the randomness used by $\mathcal{A}$
\end{itemize}
\end{definition}

\begin{definition}[Second preimage resistant]\label{def:second_preimage_resistant}
  A family of functions $\{f_k : D \rightarrow R\}_{k \in K}$ is \textbf{second preimage resistant} if:
  \begin{itemize}
\item There exists a PPT algorithm that can compute $f_k(x)$ for any index function $k$, outcome of the PPT parameter-generation algorithm \text{Gen} and any input $x \in D$;
  \item For any PPT algorithm $\mathcal{A}$, given an input $x$, it can find a different input $x'$ such that $f_k(x) = f_k(x')$ with at most negligible probability over the choice of $k$: \\
  $ \underset{\substack{k \leftarrow Gen(1^n) \\ x \leftarrow D \\ rc \leftarrow \{0, 1\}^{*}}}
  \Pr [\cA(k, x) = x' \text{such that } x \neq x' \text{ and }
    f_k(x) = f_k(x')] \leq \negl$\\
  where $rc$ is the randomness of $\mathcal{A}$;
  \end{itemize}
\end{definition}

\begin{definition}[Collision resistant]\label{def:collision_resistant}
  A family of functions $\{f_k : D \rightarrow R\}_{k \in K}$ is \textbf{collision resistant} if:
  \begin{itemize}
  \item There exists a PPT algorithm that can compute $f_k(x)$ for any index function $k$, outcome of the PPT parameter-generation algorithm \text{Gen} and any input $x \in D$;
  \item Any PPT algorithm $\cA$ can find two inputs $x \neq x'$ such that $f_k(x) = f_k(x')$ with at most negligible probability over the choice of $k$: \\
  $\underset{
      \substack{k \leftarrow Gen(1^n) \\ rc \leftarrow \{0, 1\}^{*}}}
    \Pr [\cA(k) = (x,x') \text{such that } x \neq x' \text{ and } f_k(x) =
  f_k(x')] \leq \negl$\\
where $rc$ is the randomness of $\cA$ ($rc$ will be omitted from now).
\end{itemize}

\end{definition}

\begin{theorem}\cite{Lindell}\label{thm:resistant}
Any function that is \textit{collision resistant} is also \textit{second preimage resistant}.
\end{theorem}

\begin{definition}[k-regular]\label{def:k_regular}
  
A deterministic function $f \colon D \rightarrow R$ is \textbf{k-regular} if $ \, \, \forall y \in \Im(f)$, we have ${|f^{-1}(y)| = k}$.
\end{definition}

\begin{definition}[Trapdoor Function]\label{def:trapdoor_function} 
  A family of functions $\{f_k : D \rightarrow R \}$
   is a \textbf{trapdoor function} if:
\begin{itemize}
\item There exists a PPT algorithm {\tt Gen} which on input $1^n$ outputs $(k, t_k)$, where $k$ represents the index of the function;
\item $\{f_k : D \rightarrow R\}_{k \in K}$ is a family of one-way functions;
\item There exists a PPT algorithm {\tt Inv}, which on input $t_k$ (which is called the trapdoor information) output by {\tt Gen}($1^n$) and $y = f_k(x)$ can invert $y$ (by returning all preimages of $y$\footnote{While in the standard definition of trapdoor functions it suffices for the inversion algorithm {\tt Inv} to return one of the preimages of any output of the function, in our case we require a two-regular tradpdoor function where the inversion procedure returns both preimages for any function output.})
with non-negligible probability over the choice of $(k, t_k)$ and uniform choice of $x$.

\end{itemize}
\end{definition}

\begin{definition}[Hard-core Predicate]\label{def:hardcore_predicate}
  A function $hc \colon D \rightarrow \{0, 1\}$ is a \textbf{hard-core predicate} for a function $f$ if:

  \begin{itemize}
\item There exists a QPT algorithm that for any input $x$ can compute $hc(x)$;
\item Any PPT algorithm $\mathcal{A}$ when given $f(x)$, can compute $hc(x)$ with negligible better than $1/2$ probability: \\
$ \underset{\substack{x \leftarrow D(n) \\ rc \leftarrow \{0, 1\}^{*}}} \Pr [\mathcal{A}(f(x), 1^n) = hc(x)] \leq \frac{1}{2} + \negl$, where $rc$ represents the randomness used by $\mathcal{A}$;
\end{itemize}
\end{definition}

\begin{definition}[Hard-core Function]\label{def:hardcore_function}
A function $h : D \rightarrow E$ is a \textbf{hard-core function} for a function $f$ if:
\begin{itemize}
\item There exists a QPT algorithm that can compute $h(x)$ for any input $x$
\item For any PPT algorithm $\cA$ when given $f(x)$, $\cA$ can distinguish between $h(x)$ and a uniformly distributed element in $E$ with at most negligible probability: \\
  \[ \big|
    \underset{
      \substack{x \leftarrow D(n)}
    }
    {\Pr} [\mathcal{A}(f(x), h(x)) = 1]
    -
    \underset{
      \substack{x \leftarrow D(n) \\
        r \leftarrow E(|h(x)|)}
    }
    {\Pr} [\mathcal{A}(f(x), r) = 1]\big|  \leq \negl\]
\end{itemize}
\end{definition}

The intuition behind this definition is that as far as a QPT adversary is concerned, the hard-core function appears indistinguishable from a randomly chosen element of the same length.

\begin{theorem}[Goldreich-Levin \cite{GLth}]\label{thm:GL}
  From any one-way function $f \colon D \rightarrow R$, we can construct another one-way
  function $g \colon D \times D \rightarrow R \times D $ and a hard-core predicate for $g$. If $f$ is a one-way
  function, then:
  \begin{itemize}
  \item $g(x, r) = (f(x), r)$ is a one-way function, where $|x|=|r|$.
  \item $hc(x, r) = \langle x, r\rangle \bmod 2$ is a hard-core predicate
  for $g$
  \end{itemize}
\end{theorem}
Informally, the Goldreich-Levin theorem is proving that when $f$ is a
one-way function, then $f(x)$ is hiding the xor of a random subset of
bits of $x$ from any PPT adversary\footnote{ The Goldreich-Levin proof is using a reduction from breaking the hard-core predicate $hc(x, r)$ to breaking the one-wayness of $h$. In this paper the functions we consider are one-way against quantum adversaries, and using the same reduction we conclude that $hc(x, r)$ is a hard-core predicate against QPT adversaries.}.

\begin{theorem}[Vazirani-Vazirani XOR-Condition Theorem \cite{Vazirani}]\label{thm:VV}  
  Function $h$ is hard-core function for $f$ if and only if the
  xor of any non-empty subset of $h$'s bits is a hard-core predicate for $f$.
\end{theorem}

\noindent The Learning with Errors problem (\LWE{}) can be described in the following way:
\begin{definition}[\LWE{} problem (informal)]\label{def:lwe}
Given $s$, an $n$ dimensional vector with elements in $\mathbb{Z}_q$, the task is to distinguish between a set of polynomially many noisy random linear combinations of the elements of $s$ and a set of polynomially many random numbers from $\mathbb{Z}_q$. 
\end{definition}
Regev \cite{Regev} and Peikert \cite{Peikert} have given quantum and classical reductions from the average case of \LWE{} to problems such as approximating the length of the shortest vector or the shortest independent vectors problem in the worst case, problems which are conjectured to be hard even for quantum computers.
\begin{theorem}[Reduction \LWE{}, from {{\cite[Therem 1.1]{Regev}}}]
  Let $n$, $q$ be integers and $\alpha \in (0,1)$ be such that $\alpha q > 2\sqrt{n}$. If there exists an efficient algorithm that solves $\LWE{}_{q, \bar{\Psi}_\alpha}$, then there exists an efficient quantum algorithm that approximates the decision version of the shortest vector problem \GapSVP{} and the shortest independent vectors problem \SIVP{} to within $\tilde{O}(n/\alpha)$ in the worst case.
\end{theorem}
\subsection{Quantum definitions}

We assume basic familiarity with quantum computing notions. For any function $f : A \rightarrow B$ that can be described by a polynomially-sized classical circuit, we define the controlled-unitary $U_f$, as acting in the following way:
\EQ{U_f\ket{x}\ket{y} = \ket{x}\ket{y \oplus f(x)} \, \, \, \forall x \in A \, \, \, \forall y \in B,} 
where we name the first register $\ket{x}$ control and the second register $\ket{y}$ target.
Given the classical description of this function $f$, we can always define a QPT algorithm that efficiently implements $U_f$. 

The protocol we want to implement (achieving PSRQG) can be viewed as a special case of a two-party quantum computation protocol, where one side (Client) has only classical information and thus the communication consists of classical messages. Furthermore, the client is honest, so we only need to prove security (and simulators) against adversarial server. Finally, the ideal protocol (giving same output but mediated by a trusted party; see definition below) that the real protocol implements, needs to be by itself PSRQG, i.e. obtaining the legitimate outputs should not leak any extra information (see Sections \ref{Sec:Ideal} and \ref{Sec:Privacy}).
In this paper, unless stated otherwise, we use the convention that all quantum operators considered are described by polynomially-sized quantum circuits.

We follow the notations and conventions of \cite{DNS10}. We have two parties $A,B$ with registers $\mathcal{A},\mathcal{B}$ and an extra register $\mathcal{R}$ with $\dim \mathcal{R}=(\dim\mathcal{A}+\dim\mathcal{B})$. The input state is denoted $\rho_{in}\in D(\mathcal{A}\otimes\mathcal{B}\otimes\mathcal{R})$, where $D(\mathcal{A})$ is the set of all possible quantum states in register $\mathcal{A}$. We also denote with $L(\mathcal{A})$ the set of linear mappings from $\mathcal{A}$ to itself. 
The ideal output\footnote{In case of unitary protocol $U$, while it generalises for any quantum operations.} is given by $\rho_{out}=(U\otimes\mathbb{I}_{\mathcal{R}})\cdot \rho_{in}$, where for simplicity we write $U\cdot \rho$ instead of $U\rho U^\dagger$. For two states $\rho_0,\rho_1$ we denote the trace norm distance $\Delta(\rho_0,\rho_1):=\frac12 \lVert \rho_0-\rho_1\rVert$. If $\Delta(\rho_0,\rho_1)\leq\epsilon$ then any process applied on $\rho_0$ behaves as for $\rho_1$ except with probability at most $\epsilon$.

\begin{definition}[taken from \cite{DNS10}] An $n$-step two party strategy
is denoted $\Pi^O=(A,B,O,n)$:
\begin{enumerate}
\item input spaces $\mathcal{A}_0,\mathcal{B}_0$ and memory spaces $\mathcal{A}_1,\cdots,\mathcal{A}_n$ and $\mathcal{B}_1,\cdots,\mathcal{B}_n$
\item $n$-tuple of quantum operations $(L_1^A,\cdots,L_n^A)$ and $(L_1^B,\cdots,L_n^B)$ such that $L_i^A: L(\mathcal{A}_{i-1})\rightarrow L(\mathcal{A}_i)$ and similarly for $L_i^B$.
\item $n$-tuple of global operations $(\mathcal{O}_1,\cdots,\mathcal{O}_n)$ for that step, $\mathcal{O}_i:L(\mathcal{A}_i\otimes\mathcal{B}_i)\rightarrow L(\mathcal{A}_i\otimes\mathcal{B}_i)$
\end{enumerate}
\end{definition}
The global operations (in our case) are communications that transfers some (classical) register from one party to another. The quantum state  in each step of the protocol is given by:
\EQ{\rho_1(\rho_{in})&:=&(\mathcal{O}_1\otimes\mathbb{I})(L_1^A\otimes L_1^B\otimes\mathbb{I})(\rho_{in})\nonumber\\
\rho_{i+1}(\rho_{in})&:=&(\mathcal{O}_{i+1}\otimes\mathbb{I})(L_{i+1}^A\otimes L_{i+1}^B\otimes\mathbb{I})(\rho_i(\rho_{in}))
}

\begin{definition}[Ideal Protocol]\label{def:ideal}
Given a real protocol, we call the corresponding\emph{``ideal protocol''} a protocol that has same input/output distributions with an honest run of the real protocol, but all intermediate steps are completed by a trusted third party.
\end{definition}

The security definitions are based on the corresponding ideal protocol of secure two-party quantum computation
(S2PQC) that takes a joint input $\rho_{in}\in \mathcal{A}_0\otimes\mathcal{B}_0$, obtains the state $U\cdot\rho_{in}$ and returns to each party their corresponding quantum registers. A protocol $\Pi^O_U$ implements the protocol securely, if no possible adversary in any step of the protocol, can
distinguish with a non negligible probability whether they interact with the real protocol or with a simulator (which has access to the ideal protocol). When a party is malicious we add the notation ``$\sim$'', e.g. $\tilde A$.

\begin{definition}[Simulator]
$\mathcal{S}(\tilde{A})=\langle (\mathcal{S}_1,\cdots,\mathcal{S}_n),q \rangle$ is a simulator for adversary $\tilde{A}$ in $\Pi^O_U$ if it consists of:
\begin{enumerate}
\item operations where $\mathcal{S}_i:L(\mathcal{A}_0)\rightarrow L(\tilde{\mathcal{A}_i})$ are described by polynomially-sized quantum circuits,
\item  sequence of bits $q\in\{0,1\}^n$ determining if the simulator calls the ideal functionality at step $i$ ($q_i=1$ calls the ideal functionality).
\end{enumerate}
\end{definition}
Given input $\rho_{in}$ the simulated view for step $i$ is defined as:
\EQ{\label{eq:simulator_defn}\nu_i(\tilde A,\rho_{in}):=\Tr_{\mathcal{B}_0}\left((\mathcal{S}_i\otimes\mathbb{I})(U^{q_i}\otimes\mathbb{I})\cdot \rho_{in}\right)
}
\begin{definition}[Privacy with respect to the Ideal Protocol] \label{def:private}We say that the protocol is $\delta$-private (with respect to an ideal protocol) if for all adversaries and for all steps $i$: 
\EQ{\label{eq:private_real_simul}
\Delta(\nu_i(\tilde{A},\rho_{in}),\Tr_{\mathcal{B}_i}(\tilde{\rho}_i(\tilde{A},\rho_{in})))\leq\delta
}
where $\tilde{\rho}_i(\tilde{A},\rho_{in})$ is the state of the real protocol with corrupted party $\tilde{A}$, at step $i$. 
\end{definition}

The honest-but-curious (HBC) adversaries, follow the protocol honestly, keeping records of all communication
and attempt to learn from those more than what they should. Since quantum states cannot be copied, in \cite{DNS10} they defined an adversary that could be considered the quantum analogue, called specious adversary. 

\begin{definition}[Specious]\label{def:specious}
An adversary $\tilde{A}$ is $\epsilon$-specious if there exists a sequence of operations $(\mathcal{T}_1,\cdots,\mathcal{T}_n)$, where $\mathcal{T}_i: L(\tilde{\mathcal{A}}_i)\rightarrow L(\mathcal{A}_i)$ can be described by polynomially-sized quantum circuits, such that for all $i$:
\EQ{\label{eq:specious}
\Delta\left((\mathcal{T}_i\otimes\mathbb{I})(\tilde\rho_i(\tilde A,\rho_{in})),\rho_i(\rho_{in})\right)\leq\epsilon
}
\end{definition}

In our protocol, where communications are classical, it is sensible to define a weaker version of the adversary:

\begin{definition}[Quantum-Honest-But-Curious (QHBC)]\label{def:QHBC}
An adversary $\tilde{A}$ is QHBC if it is $0$-specious.
\end{definition}

\section{Ideal Functionality}\label{Sec:Ideal}

In many distributed protocols the required communication consists of sending sequence of single qubits prepared in random states
that are unknown to the receiver (and any other third parties). What we want to achieve is a way to generate remotely single
qubits that are random and (appear to be) unknown to all parties but the ``client'' that gives the instructions.

In this work, for clarity and having in mind the applications we wish to implement, we will focus on a particular choice for the set $R$ of possible states that contains eight different single-qubit states (see below). One could easily modify our work to restrict to a smaller set (e.g. the four BB84 states \cite{BB84} that would actually simplify our proofs) or a larger set.

\begin{definition}
Let $\ket{+_\theta}=1/\sqrt{2}\left(\ket{0}+e^{i\theta}\ket{1}\right)$. We define the set of states 

\EQ{\label{eq:output_states}
R:= \{ \ket{+_\theta} \} \ \textrm{ where } \ \theta\in\{ 0,\pi/4,\pi/2,\cdots,7\pi/4\}
}
\end{definition}
By including magic states ($\ket{+_{\pi/4}}$), this set of states can be viewed as a ``universal'' resource, as applying Clifford operations on those states is sufficient for universal quantum computation. Furthermore, it is sufficient to implement both Blind Quantum Computation (e.g. \cite{bfk}) and Verifiable Blind Quantum Computation (e.g. \cite{FKD2017}).

\begin{algorithm}[H]
\caption{Ideal Functionality: Secret Random Qubit Generator (SRQG) -- $\mathcal{F}(M,p)$}
\label{ideal:q_factory}
\textbf{Public Information:} A distribution on pairs of lists $M$, intuitively containing the values of the classical variables used by the client and by the server\\
\textbf{Trusted Party:}\\
-- With some probability $p$ returns to both parties $\mathsf{abort}$, otherwise:\\
-- Samples $(m_C, m_S)\leftarrow M$\\
-- Samples $r\leftarrow\{0,1\}^3$\\
-- Prepares a qubit in state $\ket{+_{(r\pi/4)}}$\\
\textbf{Outputs:}\\
% -- Either returns $(m, \mathsf{abort})$ to both client and server\\
-- Either returns $\mathsf{abort}$ to both client and server\\
-- Or returns $(m_C, r)$ to the client, and $(m_S, \ket{+_{(r\pi/4)}})$ to the server
% -- Or returns $r_m$ to the client, and $\ket{+_{(r_m\pi/4)}}$ to the server
\end{algorithm}
\noindent\textbf{Remarks}: (i) The outcome of this functionality is the client ``sending'' the qubit $\ket{+_\theta}$ (that he knows) to the server, thus simulating a quantum channel. (ii) We note that there is an abort possibility and some auxiliary classical message $m$, both included to make the functionality general enough to allow for our construction. Furthermore, the classical description of the qubit $r$ and the classical message $m$ are totally uncorrelated (as $r$ is chosen randomly for each $m$). (iii) While the server \emph{can} learn something about the classical description (e.g. by measuring the qubit), this information is limited and is the exact same information that he could obtain if the client had prepared and send a random qubit. Therefore, the privacy is defined with respect to this ideal setting.

We are interested only in the honest-but-curious setting for now. The idea is that we will allow the adversary to have access to the classical registers/variables of the server (we will call these information a ``\emph{view}''), as well as the classical variables produced by the ideal functionality (uncorrelated with the quantum output, so secure by definition). The goal of the adversary will be to distinguish whether he is interacting with a view of the ideal functionality or a view of the real protocol. More formally, we will denote by $\cP_S$ the view of server $S$ in protocol $\cP$, which is the list of the content of the variables/classical registers assigned by the server $S$ in the protocol $\cP$. Similarly, $\cF_S$ will be the view of the server $S$ in the ideal functionality $\cF$, equal to the value of $m_S$ in a run of the idea functionality.

\begin{definition}[Pseudo-Secret Random Qubit Generator (PSRQG)]\label{def:c_q_factory}
  We call a protocol $\mathcal{P}(1^n)$ to be
  $\eps(n)$-Pseudo-Secret Random Qubit Generator ($\eps(n)$-PSRQG) if there exists a SRQG $\cF(M,p)$ such that for all
  Quantum Polynomial Time (QPT) adversaries/distinguishers
  $\cA$:
  \EQ{\label{eq:c_q_factory}\left|\Pr[\mathcal{A}(\mathcal{P}_S(1^n))=1]-\Pr[\mathcal{A}(\cF_S(M,p))=1]\right|\leq \eps(n)}
  If $\eps(n)$ is a negligible function, we will simply denote it
PSQRG (omitting $\eps(n)$).
\end{definition}

To achieve the PSRQG functionality we define an ideal protocol, called Ideal QFactory\footnote{We call this a ``qubit factory'', since we use this protocol to produce strings of qubits.}, mediated
by a trusted third party that (under certain assumptions) achieves the PSRQG functionality. 
This ideal protocol, can be realised by a concrete protocol without any trusted parties (see later), and certain choices in the definition of the ideal QFactory (e.g. the function required) are done with this in mind.

\begin{algorithm}[H]
\caption{Ideal QFactory Protocol}
\label{ideal:c_q_factory}
\textbf{Public Information:} A security parameter $n\in \mathbb{N}^*$, a trapdoor one-way function that are quantum-safe, two-regular and collision resistant (or the weaker second preimage resistance property, see \autoref{rmk:collision_resistance}) $\{f_k \colon D \rightarrow R\}_{k \in K}$ and a family of functions $\{g_k \colon D \times D \times E \rightarrow \{0,1\}^3\}_{k \in K}$\\
\textbf{Trusted Party:}\\
-- Runs the algorithm $\mathsf{Gen}(1^n)=(k,t_k)$ of the trapdoor function\\
-- Samples randomly $x \leftarrow D,\beta\leftarrow E$\\
-- Using $t_k$, computes the unique other preimage $x' \neq x$ such that $f_k(x)=f_k(x')=y$\\
-- If the last bit of $x$ and $x'$ is the same, $\mathsf{abort}$ otherwise\\
-- Computes  $\tilde{B} :=g_k(x,x',\beta)$. Setting $\theta := \tilde{B} \times \pi/4$, prepares a qubit in the state $\ket{+_\theta}$\\
\textbf{Outputs:}\\
-- Either returns $\mathsf{abort}$ to both parties\\ 
-- Or returns $(k,y,\beta, \ket{+_\theta})$ to server $S$ and $(t_k,y,\beta,\theta)$ to client $C$. Note that the $\theta$ is optional and could have been recomputed by the client from $t_k$.
\end{algorithm}

\begin{remark}\label{rmk:collision_resistance}
  It appears that the second preimage resistance property will be enough to prove the security of our scheme in the honest-but-curious setting. However, as soon as the server can be malicious, the collision resistance property will be very important, else the server might forge known valid states, which would break the security.
\end{remark}

We will denote by $M_{QF}$ the distribution obtained by sampling as above the index $k$ and trapdoor $t_k$ according to $\mathsf{Gen}(1^n)$, the $y$ uniformly in the elements of $R$ having two preimages, and the $\beta$ uniformly in $E$, and then outputting $((t_k,y,\beta),(k,y,\beta))$.

\begin{lemma}\label{lemma:hardcore}
  Ideal QFactory \autoref{ideal:c_q_factory} is a PSRQG protocol as described in  \autoref{def:c_q_factory} (with $M$ having the distribution $M_{QF}$) if the function
  $g_k(x,x',\beta)$ (restricted on $x,x'$ such that $f_k(x)=f_k(x')$) is a hard-core function for $f_k$. 
\end{lemma}

\begin{proof}
We can see that \autoref{ideal:c_q_factory} is identical with \autoref{ideal:q_factory} with $M=M_{QF}$ (since the Client having $t_k$ can determine if it aborts or not), apart from the fact that in \autoref{ideal:c_q_factory} the state received by the server is $\ket{+_\theta}$, while in \autoref{ideal:q_factory} is $\ket{+_{r}}$. 

Now we use the fact that $g_k$ is a hard-core function. By definition \ref{def:hardcore_function}, for a QPT adversary that  has access to $m=(k,y=f_k(x),\beta)$ the value of the hard-core function  $g_k(x,x',\beta)=4\theta/\pi$ where $x,x'$ are the unique preimages of $y$, is indistinguishable (up to negligible probability) to that of a random value $r$. It follows that such adversary cannot distinguish (apart with negligible probability) whether he received  the state $\ket{+_\theta}$ as in \autoref{ideal:c_q_factory} or the state $\ket{+_{r}}$ as in the ideal functionality described on \autoref{ideal:q_factory}, and therefore satisfies Eq. \ref{eq:c_q_factory}).
\end{proof}

It is not sufficient to prove that given the image $y=f_k(x)$ it is hard to obtain the exact value of the function $g$ (we will omit the $k$ if it is clear from the context), we want the stronger requirement that given $y$, 
a QPT adversary obtains no advantage in distinguishing the value of $g$ (the classical description of the state), from a totally random value $r$. Intuitively, what \autoref{ideal:c_q_factory} achieves, is that it produces (truly) random qubits in states that are pseudo-secret, i.e. their classical description is computationally unknown to anyone that does not have access to the trapdoor $t_k$ (i.e. the server).

\section{The Real Protocol}\label{Sec:Protocol}

We assume the existence\footnote{See \autoref{Sec:Functions} for  our function. With that choice, we are guaranteed that the last bits of the two preimages are always different, and thus no need for an abort. We keep the protocol general so that different functions can be used.} 
of a family
$\{f_k \colon \{0,1\}^n \rightarrow \{0,1\}^m\}_{k \in K}$ of trapdoor
one-way functions that are two-regular and collision resistant (or the weaker second preimage resistance property, see \autoref{rmk:collision_resistance}) even against a quantum adversary. For
any $y$, we will denote by $x(y)$ and $x'(y)$ the two unique different
preimages of $y$ by $f_k$ (if the $y$ is clear, we may remove it).
Note that because of the two-regularity property $m \geq n-1$. We use subscripts to denote the different bits of the strings.

\begin{algorithm}[H]
\caption{Real QFactory Protocol}\label{protocol:concrete_c_q_factory}
\textbf{Requirements:} \\
Public: A family $\mathcal{F}=\{f_k : \{0,1\}^n \rightarrow \{0,1\}^m \}$ of trapdoor one-way functions that are quantum-safe, two-regular and collision resistant (or second preimage resistant, see \autoref{rmk:collision_resistance})\\
\textbf{Input:}\\
-- Client: uniformly samples a set of random three-bits strings $\alpha=(\alpha_1,\cdots,\alpha_{n-1})$ where $\alpha_i\leftarrow\{0,1\}^3$, and runs the algorithm $(k,t_k) \leftarrow \text{Gen}_{\mathcal{F}}(1^n)$. The $\alpha$ and $k$ are public inputs (known to both parties), while $t_k$ is the ``private'' input of the Client.\\
\textbf{Stage 1: Preimages superposition} \\
-- Client: instructs Server to prepare one register at $\otimes^n H\ket{0}$ and second register initiated at $\ket{0}^{m}$

-- Client: returns $k$ to Server and the Server applies $U_{f_k}$ using the first register as control and the second as target

-- Server: measures the second register in the computational basis, obtains the outcome $y$ and returns this result $y$ to the Client. Here, an honest Server would have a state ${(\ket{x}+\ket{x'}) \otimes \ket{y}}$ with $f_k(x)=f_k(x')=y$ and $y\in \Im f_k$.
\\
\textbf{Stage 2: Squeezing}\\
-- Client: instructs the Server to measure all the qubits (except the last one) of the first register in the $\left\{\ket{0}\pm e^{\alpha_i\pi/4}\ket{1}\right\}$ basis. Server obtains the outcomes $b=(b_1,\cdots,b_{n-1})$ and returns the result $b$ to the Client
\\
-- Client: using the trapdoor $t_k$ computes $x,x'$. Then check if the $n$th bit of $x$ and $x'$ (corresponding to the $y$ received in stage 1) are the same or
different. If they are the same, returns $\mathsf{abort}$, otherwise, obtains the classical
description of the Server's state.\\
\textbf{Output:} If the protocol is run honestly, when there is no
abort, the state that Server has is $\ket{+_\theta}$, where the Client (only)
knows the classical description (see \autoref{Thm:correctness}):

\EQ{\label{eq:theta} \theta=\frac{\pi}{4}(-1)^{x_n}\sum_{i=1}^{n-1}(x_i-x_i')(4b_i+\alpha_i)\bmod 8}
\end{algorithm}

\noindent\textbf{Remarks:} The first thing to note is that the server
should not only be unable to guess $\theta$ from his classical
communications, but he should also be unable to distinguish it from a
random string with probability greater than negligible. We will prove
this later, but for now it is enough to point out that $\theta$
depends on the pre-images $x$ and $x'$ of $y$ (which the Client can
obtain using $t_k$).

The second thing to note is that previously, in \autoref{ideal:c_q_factory} and in \autoref{lemma:hardcore}, we used the variable $\beta$. In our
case, $\beta$ corresponds to both $\alpha_i$'s and $b$. While our
expression resembles the inner product in the Goldreich-Levin (GL)
theorem, it differs in a number of places and our proof (that $\theta$
is a hard-core function), while it builds on GL theorem proof, is
considerably more complicated. Details can be found in the security
proof, but here we simply mention the differences: (i) our case involves three-bits rather than a predicate, and the different bits, if we view them separately, may not be independent, (ii) we have a term $(x-x')$ rather than a single preimage, so rather than the one-way property of the function we will need the second preimage resistance and (iii) for the same reason, if we view our function as an inner product, it can take both negative and positive values ($(x-x')$
could be negative).

A third thing to note is that we have singled-out the last qubit of the first register, as the qubit that will be the output qubit. One could have a more general protocol where the output qubit is chosen randomly, or, for example, in the set of the qubits that are known to have different bit values between $x$ and $x'$, but this would not
improve our analysis so we keep it like this for simplicity. Moreover,
while the ``inner product'' normally involves the full string $x$ that
one tries to invert, in our case, it does not include one of the bits
(the last) of the string we wish to invert. It is important to note,
that it does not change anything to our proofs, since if one can
invert all the string apart from one bit with inverse polynomial
probability of success, then trivially one can invert the full string
with inverse polynomial probability (by randomly guessing the
remaining bit or by trying out both values of that bit). Therefore all
the proofs by contradiction are still valid and in the remaining, for
notational simplicity, we will take the inner products to involve all
$n$ bits.

\subsection{Correctness and intuition}

\begin{theorem}\label{Thm:correctness}
If both the Client and the Server follow \autoref{protocol:concrete_c_q_factory}, the protocol aborts when $x_n=x_n'=f_k^{-1}(y)$, while otherwise the Server ends up with the output (single) qubit being in the state $\ket{+_\theta}$, where $\theta$ is given by Eq. (\ref{eq:theta}).
\end{theorem}

\begin{proof}
In the first stage, before the first measurement, but after the application of $U_{f_k}$, the state is $\sum_x\ket{x}\otimes\ket{f_k(x)}$. What the
measurement does, is that it collapses the first register in the equal
superposition of the two unique preimages of the measured $y=f_k(x)=f_k(x')$,
in other words in the state
$\left(\ket{x}+\ket{x'}\right)\otimes\ket{y}$. It is not possible, even for malicious adversary (not considered here), to force the output of the measurement to be a given $y$ (see \cite{postbqp} for relation of PostBQP with BQP).
This completes the first stage of the protocol. Before proceeding with the proof of correctness we make three observations.

By the second preimage resistance property of the trapdoor function, learning $x$ is not sufficient to learn $x'$ but with negligible probability, and intuitively, by the stronger collision resistance property, even a malicious server cannot forge a state $\ket{x}+\ket{x'}$ (with $f(x)=f(x')$) fully known to him.

Then, we examine what happens if the last bit of $x$ and $x'$ are the same and see why the protocol aborts. In this case, in the first register, the last qubit is in product form with the remaining state, and therefore any further measurements in stage 2
do not affect it, leaving it in the state $\ket{x_{n}}$. Because of this, the output state is not of the form of Eq. (\ref{eq:theta}), while including this states in the set of possible outputs would change considerably our analysis.

Finally, we should note that the resulting state is essentially a Greenberger-Horne-Zeilinger (GHZ)
state \cite{GHZ1989}: let $G$ be the set of bits positions where $x$ and $x'$ differ (which
include $n$ -- output qubit), while $\bar{G}$ is the set where they are identical. The state is then (where we no longer keep the qubits in order, but group them depending on their belonging to $G$ or
$\bar{G}$):
\EQ{
\big(\otimes_{i\in \bar{G}}\ket{x_{i}}\big)\otimes\big(\otimes_{j\in G}\ket{x_{j}}+\otimes_{j\in G}\ket{x_{j}\oplus 1}\big)
}
This can be rewritten as (up to trivial re-normalization):
\EQ{
\big(\otimes_{i\in \bar{G}}\ket{x_{i}}\big)\otimes\Big(\prod_{j\in G}X^{x_{j}}\Big)\big(\ket{0\cdots 0}_{G}+\ket{1\cdots 1}_{G}\big)
}

It is now evident that the state at the end of Stage 1 is a tensor
product of isolated $\ket{0}$ and $\ket{1}$ states, and a GHZ state
with random $X$'s applied. You can find on
Figure~\ref{fig:drawing_ghz} an illustration of this state\footnote{GHZ-states when viewed as graph states correspond to stars due to the corresponding graph as we can see here too.}, before and
after the Stage 2.

\begin{figure}
\centering
	\begin{subfigure}[t]{0.47\textwidth}
  \centering
  \includegraphics[width=\linewidth]{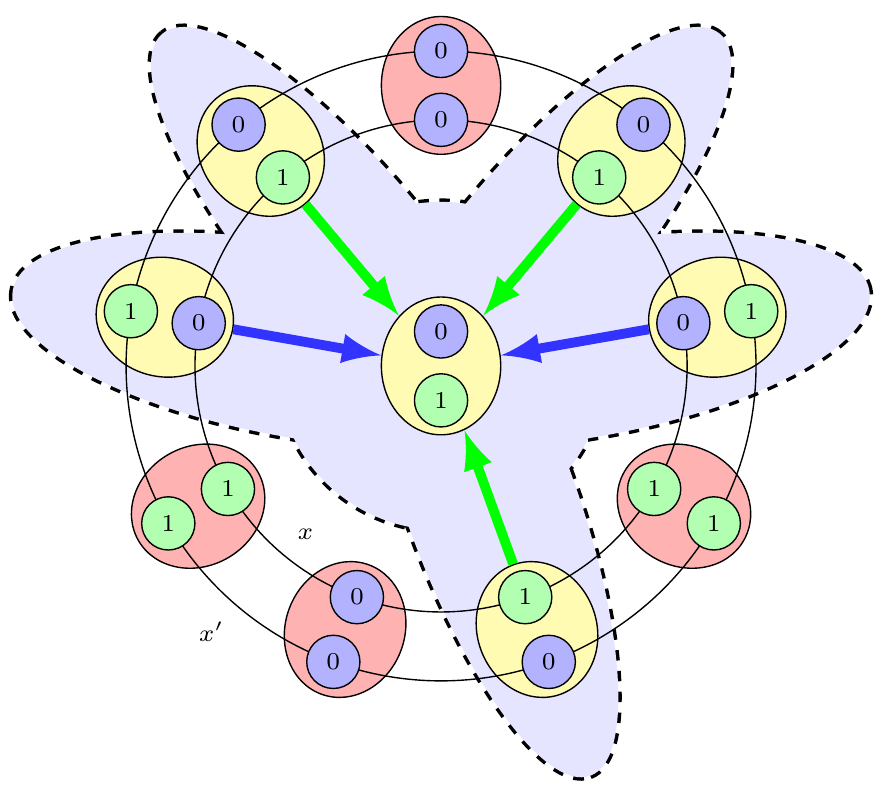}
  \caption{End of Stage 1: The yellow qubits are in a big \mbox{GHZ-like} state. The server does not know which qubits are in the GHZ state, and which qubits are not (in red).}
  \label{fig:drawing_ghz_left}
\end{subfigure}\hspace*{0.05\textwidth}
\begin{subfigure}[t]{0.47\textwidth}
  \centering
  \includegraphics[width=\linewidth]{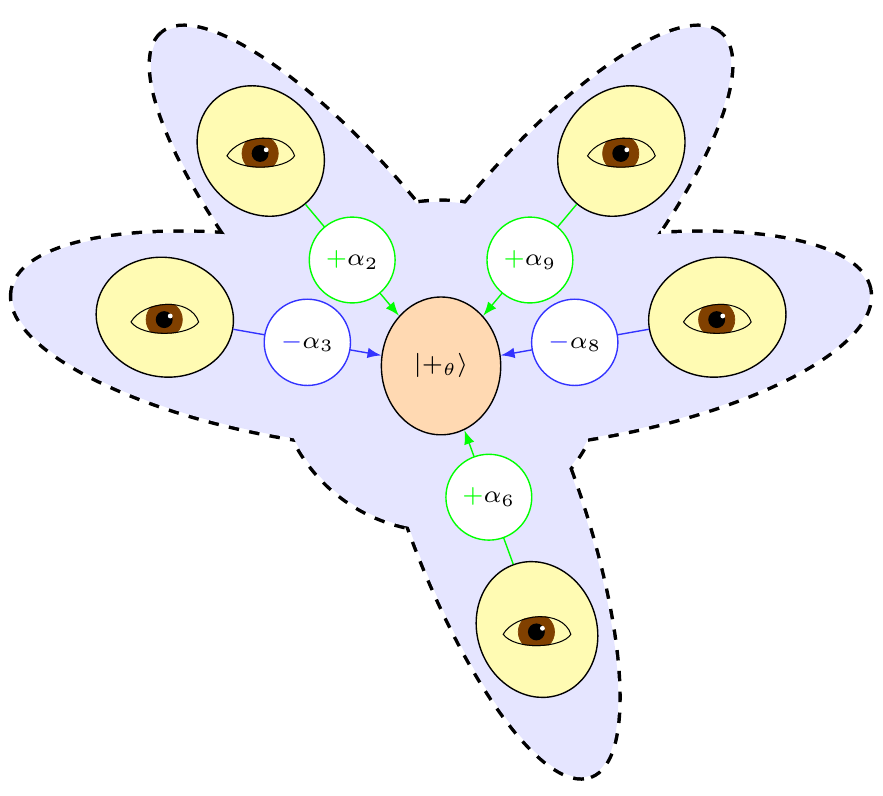}
  \caption{End of Stage 2: Measuring any qubit in the GHZ state will rotate the last (ouput) qubit depending on the angle (and result) of the measurement.}
  \label{fig:drawing_ghz_right}
\end{subfigure}
\caption{A simplified representation of the protocol. The red and yellow ellipses represent the qubits, the inner circle contains the bits of $x$ and the outer circle contains the bits of $x'$. The central qubit is the last one, which is not measured and which will be the output qubit.}
\label{fig:drawing_ghz}
\end{figure}

The important thing to note, is that the set $G$, that determines which qubits are in the GHZ state and which qubits are not, is \emph{not} known to the server (apart from the fact that the position of the output qubit belongs to $G$ since otherwise the protocol aborts). Moreover, this set denotes the positions where $x$ and $x'$ differ, which is given by the XOR of the two preimages $x\oplus x':=(x_1\oplus x_1',\cdots,x_n\oplus x_n')$. Because of
second preimage resistance of the function, the server should not be able to invert and obtain $x\oplus x'$ apart with negligible probability (without access to the trapdoor $t_k$). This in itself does not guarantee that the Server cannot learn \emph{any} information about the XOR of the preimages, but we will see that the actual form of the state is such that being able to obtain information would lead to invert the full XOR and thus break the second preimage resistance.

Now let us continue towards Stage 2. Measuring a qubit (other than the
last one) in $\bar{G}$ has no effect on the last qubit (since it is
disentangled). When the qubit index is in $G$, then measuring it at angle
$\alpha_i\pi/4$ gives a phase to the output qubit of the form
$(-(-1)^{x_{i}}\alpha_i+4b_i)\pi/4$ as one can easily
check\footnote{The $(-1)^{x_i}$-term arises because of the commutation
  of $X_i^{x_{i}}$ with the measurement angle, and the final
  $X_n^{x_n}$ gate gives an overall $(-1)^{x_n}$ to the angle of
  deviation}.
Therefore, adding all the phases leads to the output state being:
\EQ{
\ket{+_\theta}; \ \theta=\frac\pi4(-1)^{x_{n}}\left(\sum_{i\in G\setminus\{n\}}\big(-\alpha_i(-1)^{x_{i}}+4b_i\big)\right)\bmod 8
}
Because $\theta$ is defined modulo $2\pi$ and $-4 = 4 \bmod 8$, we can express the output angle in a more
symmetrical way:
\EQ{
\theta=\frac\pi4(-1)^{x_{n}}\left(\sum_{i=1}^{n-1}(x_{i}-x'_{i})\big(4b_i+\alpha_i\big)\right) \bmod 8
}
Note that because the angles are defined modulo $2\pi$, one can
represent this angle as a 3-bits string $\tilde{B}$ (interpretable as
an integer) such that $\theta := \tilde{B} \times \frac{\pi}{4}$ and
eventually remove the $(-1)^{x_{n}}$ if needed by choosing the suitable convention in defining $x$ and
$x'$.
\end{proof}

A final remark is that in an honest run of this protocol, the measurement outcomes $b_i$ and $y$ are uniformly chosen from $\{0,1\}$ and $\Im(f_k)$  respectively. This justifies why in the honest-but-curious model we can view the protocol as sampling randomly
the different $\alpha,y,b$'s.

\section{Privacy against QHBC adversaries}\label{Sec:Privacy}

Here we will prove the security of \autoref{protocol:concrete_c_q_factory} against QHBC adversaries (\autoref{def:QHBC}). It can easily be generalised for specious adversaries (\autoref{def:specious}). Before proceeding further, it is worth stressing that this security level has three-fold importance. Firstly, the QHBC model concerns any application of PSQRG that involves a protocol where the adversaries are third parties that have access to the classical communication and nothing else. In this case, we can safely assume that the quantum part of the protocol is followed honestly and we only require to prove that the third parties learn nothing about the classical description of the state from the classical public communication.
Second case of interest is scenarios where the ``server'' does not intend to sabotage/corrupt the computation but may be interested to learn (for free) extra information. In such case, the protocol should be followed honestly, since any non-reversible deviation other than copying classical information could corrupt the computation. Finally, the QHBC case, as in the classical setting, is a first step towards proving the full security against malicious adversaries, as we will discuss in \autoref{Sec:Malicious}).

\begin{theorem}\label{Thm:privacy}
\autoref{protocol:concrete_c_q_factory} realises a PSRQG Ideal Protocol (as in \autoref{def:c_q_factory}) that is private with respect to this ideal protocol (as in \autoref{def:private}) against a QHBC server $\mathcal{A}$ (\autoref{def:QHBC}).
\end{theorem}

Before proving the privacy with respect to the ideal functionality (see below for construction of simulators), the first step is to show that the corresponding ideal protocol (\autoref{def:ideal}) is a PSRQG. By \autoref{lemma:hardcore} this reduces to proving that the classical description is a hard-core function with respect to $f_k$.

\begin{theorem}\label{Thm:hardcore}
The function $\theta$ given here
\EQ{\theta=\frac{\pi}{4}\left(\sum_{i=1}^{n-1}(x_i-x_i')(4b_i+\alpha_i)\right) \bmod 8}
as was defined in \autoref{protocol:concrete_c_q_factory}, is a hard-core function with respect to $f_k$.\\
NB: here the collision resistance is not needed and is replaced by the weaker second preimage resistance property.
\end{theorem}

\begin{proof}[Sketch Proof of Theorem \ref{Thm:hardcore}]
In \autoref{protocol:concrete_c_q_factory}, the adversary (Server) can only use  the classical information that he possesses ($k,y,\alpha,b$) in  order to try and guess with some probability the value of $\theta$  in the case that there is no abort. Since the adversary follows the honest protocol, the choices of $y,b$ are truly random (and not  determined by the adversary as he could in the malicious case).

\noindent\textbf{Outline of sketch proof:} We first express the classical description of the state into expressions for each of the  corresponding three bits. The aim is to prove that it is impossible to distinguish the sequence of these three bits from three random bits with non-negligible probability. To show this we follow five steps. In \textbf{Step 1} we express each of the the bits as a sum
  mod two, of an inner product (of the form present in GL theorem) and some other terms. In \textbf{Step 2} we show that guessing the sum  modulo two of the two preimages breaks the second preimage resistance of the function and thus is impossible. We assume that the  adversary can achieve some inverse polynomial advantage in guessing  certain predicates and in the remaining steps we show that in that  case he can obtain a polynomial inversion algorithm for the one-way function $f_k$, and thus reach the contradiction. In \textbf{Step 3} we use the Vazirani-Vazirani \autoref{thm:VV} to reduce the proof of hard-core function to a number of single hard-core bits (predicates). In \textbf{Step 4} we use a Lemma that allows us  to fix all but one variable in each expression, with an extra cost
that is an inverse polynomial probability and therefore the (fixed variables) guessing algorithm still needs to have negligible success  probability. Finally, in \textbf{Step 5}, we reduce all the  predicates in a form of a known hard-core predicate XOR with a function that involves variables not included in that predicate. Using the previous step, it reduces to guessing the XOR of a hard-core predicate with a constant, which is bounded by the probability of guessing the (known to be hard-core) predicate.

Here we give the sketch described above, while the full proof can be found in the \autoref{app:hardcore}. Let us start by defining

\EQ{\label{eq:proof1}\widetilde{B}=\widetilde{B}_1\widetilde{B}_2\widetilde{B}_3=\left(\sum_{i=1}^{n-1}(x_i-x_i')(4b_i+\alpha_i)\right)\bmod 8}
where $\widetilde{B}_i$ are single bits. Moreover, we treat $x,x'$ as vectors in $\{0,1\}^n$; we define $\alpha^{(j)}=(\alpha_1^{(j)},\cdots,\alpha_{n-1}^{(j)})$ the vector that involves the $j\in\{1,2,3\}$ bit of each of three-bit strings $\alpha$, and we define $\tilde{x}:=x\oplus x'$. We define $z$ as a vector in $\{-1,0,1\}^n$ defined as the element-wise differences of the bits of $x$ and $x'$, i.e. $z_i=x_i-x_i'$. Finally, as in GL theorem, we use the notation for the inner product $\langle a,b\rangle=\sum_{i=1}^{n-1} a_ib_i$.

We will prove that any QPT adversary $\mathcal{A}$ having all the classical information that Server has $(y,\alpha,b)$, can guess $\widetilde{B}$ with at most negligible probability 

\EQ{\underset{\substack{ x \leftarrow \{0, 1\}^n \\ \alpha \leftarrow \{0, 1\}^{3n} \\ b \leftarrow \{0,1\}^n }} {\Pr} [\mathcal{A}(f(x), {\alpha}^{(1)}, {\alpha}^{(2)}, {\alpha}^{(3)}, b) = \widetilde{B_1}\widetilde{B_2}\widetilde{B_3}] \leq \frac18+\negl
}
where for simplicity we denote the function $f$ instead of $f_k$. This means that the adversary $\mathcal{A}$ cannot distinguish $\widetilde{B}$
from a random three-bit string with non-negligible probability and thus
\autoref{protocol:concrete_c_q_factory} is PSRQG as given in
\autoref{def:c_q_factory}.

\noindent\textbf{Step 1:} We decompose Eq. (\ref{eq:proof1}) into
three separate bits, and use the variable $\tilde{x},z$ defined above.

\EQ{\label{eq:proof2}
\widetilde{B}_3&=&\langle\tilde{x},\alpha^{(3)}\rangle\bmod 2 \nonumber\\
\widetilde{B}_2&=& \langle\tilde{x},\alpha^{(2)}\rangle\bmod 2\oplus h_2(z,\alpha^{(3)})\nonumber\\
\widetilde{B}_1&=&\langle\tilde{x},\alpha^{(1)}\rangle\bmod 2\oplus h_1(z,\alpha^{(3)},\alpha^{(2)},b)
}
where the derivation and exact expressions for the functions $h_1,h_2$
are given in \autoref{app:hardcore}. We notice from
Eq. (\ref{eq:proof2}) that each bit includes a term of the form
$\langle \tilde{x},\alpha^{(i)}\rangle\bmod 2$ which on its own is a
hard-core predicate following the GL theorem.

\noindent\textbf{Step 2:} By the second preimage resistance we have:

\EQ{\label{eq:proof3}& &\underset{x \leftarrow \{0,1\}^n} {\Pr} [\mathcal{A}(1^n, x) = x' \text{ such that } f(x) = f(x') \text{ and }x\neq x'] \leq \negl \Rightarrow \nonumber\\
& &\underset{x \leftarrow \{0,1\}^n} {\Pr} [\mathcal{A}(1^n, x) = x' \oplus x=\tilde{x}]\leq \negl
}
For each bit $j\in\{1,2,3\}$, separately we assume that the adversary can achieve an advantage in guessing the $\tilde{x}$ which is $\frac12+\eps_j(n)$. Then, similarly to GL theorem, we prove that if this $\eps_j(n)$ is inverse polynomial, this leads to contradiction with Eq. (\ref{eq:proof3}) since one can obtain an inverse-polynomial inversion algorithm for the one-way function $f$.

\noindent\textbf{Step 3:} While each bit includes terms that on its own it would make it hard-core predicate (as stated in Step 1), 
if we XOR the overall bit with other bits it could destroy this property. To proceed with the proof that $\widetilde{B}$ is hard-core function we use the Vazirani-Vazirani theorem which states that it suffices to show that individual bits as well as combinations of XOR's of individual bits are all hard-core predicates. In this way one evades the need to show explicitly that the guesses for different bits are not correlated. To proceed with the proof, we use a trick that ``disentangles'' the different variables.

\noindent\textbf{Step 4:} We would like to be able to fix one variable and vary only the remaining, while in the same time maintain some bound on the guessing probability.

The advantage $\eps_j(n)$ that we assume the adversary has for guessing one bit (or an XOR) is calculated ``on average'' over all the random choices of $(\tilde{x},\alpha^{(i)},b)$. Using \autoref{lemma:proof1} we can fix one-by-one all but one variable (applying the lemma iteratively, see \autoref{app:hardcore}). With suitable choices, the cardinality of the set of values that satisfies all these conditions is $O(2^n\eps_j(n))$ for each iteration. Unless $\eps_j(n)$ is negligible, this size is an inverse polynomial fraction of all
values. This suffices to reach the contradiction. The actual inversion probability that we will obtain is simply a product of the extra cost of fixing the variables with the standard GL inversion probability. This extra cost is exactly the ratio of the cardinality of the $\text{Good}$ sets (defined below) to the set of all values and is $O(\eps_{v_i}(n))$.

\begin{lemma}\label{lemma:proof1}
Let $\underset{(v_1,\cdots,v_k)| v_j\leftarrow\{0,1\}^n \forall j}{\Pr}[\text{Guessing}]\geq p+\eps(n)$, then for any variable $v_i$, $\exists$ a set $\text{Good}_{v_i}\subseteq\{0,1\}^n$ of size at least $\frac{\eps(n)}{2}2^n$, where $\forall$ $v_i\in \text{Good}_{v_i}$

\[
\underset{(v_1,\cdots,\cancel{v_i},\cdots,v_k)| v_j\in\{0,1\}^n}{\Pr}[\text{Guessing}]\geq p+\frac{\eps(n)}2
\]
where the probability is taken over all variables except $v_i$. 
\end{lemma}

\noindent\textbf{Step 5:} If the expression we wish to guess involves
XOR of terms that depend on different variables, then by using Step 4
we can fix the variables of all but one term. Then we note that trying
to guess a bit (that depends on some variable and has expectation value close to $1/2$) is at least as hard as
trying to guess the XOR of that bit with a constant. For example, if
the bit we want to guess is
$\langle \widetilde{x}, r_1\rangle \bmod 2 \oplus h(z, r_2, r_3)]$ and
we have a bound on the guessing probability where only $r_1$ is
varied, then we have: \footnote{Here and in the full proof, when we compare winning probabilities for QPT adversaries, it is understood that we take the adversary that maximises these probabilities.}

\EQ{&\underset{\substack{r_1 \leftarrow \{0, 1\}^{n}}} {\Pr} [\mathcal{A}(f(x), r_1, r_2, r_3) = \langle \widetilde{x}, r_1\rangle \bmod 2 \oplus h(z, r_2, r_3)]\leq \nonumber\\ & \underset{\substack{r_1 \leftarrow \{0, 1\}^{n}}} {\Pr} [\mathcal{A}(f(x), r_1, r_2, r_3) = \langle \widetilde{x}, r_1\rangle \bmod 2]\nonumber 
 } 

We note that all bits of $\widetilde{B}$ and their XOR's can be
brought in this form. Then using this, we can now prove security, as
the r.h.s. is exactly in the form where the GL theorem provides an
inversion algorithm for the one-way function $f$. For details, see \autoref{app:hardcore}.
\end{proof}

We now return and prove \autoref{Thm:privacy}.

\begin{proof}[Proof of \autoref{Thm:privacy}]
In the proof we use QHBC adversaries, but following closely the more general Specious adversaries, we can see that it would easily generalise. The proof has two steps. In the \textbf{first step} of the proof, using the operators $\mathcal{T}_i$ (as per definition of specious) and the existence of certain fixed state (see below) the simulator can reproduce the real view of the Server, if he can reproduce the honest state $\rho_i(\rho_{in})$ of the corresponding part of the protocol. The \textbf{second step} of the proof is to notice that apart from the last step of the protocol (decision to abort or not), the (only) secret input $t_k$ of the Client plays no role, and thus the simulator can reproduce the view of the Server without calling the ideal functionality. Finally, the simulator of the last step of the protocol, calls the ideal functionality (and thus $q_i=1$ in Eq. (\ref{eq:simulator_defn})) and receives the decision to abort (without access to the secret $t_k$).

\noindent \textbf{Step 1:} We use the no-extra information lemma from \cite{KW2017}:

\begin{lemma}[No-extra Information (from \cite{KW2017})]\label{rushing1}
Let $\Pi_U=(A,B,n)$ be a correct protocol for two party evaluation of $U$. Let $\tilde A$ be any $\epsilon$-specious
adversary. Then there exists an isometry $T_i:\tilde A_i\rightarrow A_i\otimes \hat{A}$ and a (fixed) mixed state $\hat{\rho_i}\in D(\hat{A_i})$ such that for all joint input states $\rho_{in}$,

\EQ{\label{eq:rushing}
\Delta\left((T_i\otimes\mathbb{I})(\tilde\rho_i(\tilde A,\rho_{in})),\hat{\rho_i}\otimes\rho_i(\rho_{in})\right)\leq 12\sqrt{2\epsilon}
}
where $\rho_i(\rho_{in})$ is the state in the honest run and $\tilde\rho_i(\tilde A,\rho_{in})$ is the real state (with the specious adversary $\tilde A$).
\end{lemma}
By setting $\epsilon=0$ (as QHBC) and using the inverse of the isometry $T_i$, we have\footnote{We denote the two parties $C$ for Client and $S$ for Server and their corresponding spaces (instead of the generic $A,B$ used in the definitions).}

\EQ{
\tilde\rho_i(\tilde S,\rho_{in})=T_i^{-1}\otimes\mathbb{I}(\hat{\rho_i}\otimes\rho_i(\rho_{in}))
}
and the operation $S_i$ of the simulator for \emph{any} step, consists of generating $\rho_i(\rho_{in})$ (see next part of the proof), tensor it with the fixed state $\hat{\rho_i}$ and apply the inverse of the isometry $T_i$. This recovers exactly the real state $\tilde\rho_i(\tilde S,\rho_{in})$ and thus tracing out the system of the Client to obtain the simulated view $\nu_i(\tilde S,\rho_{in})$ gives  $(\delta=0)$-private with respect to the ideal protocol (see Eq. (\ref{eq:private_real_simul})).

\noindent \textbf{Step 2:} We give below the honest states at the two steps of the protocol before the Server (classically) communicates with the Client, noting that a simulator (with no access to the private information $t_k$) could interact with the Server (instead of the Client) just following the normal steps of the protocol, using the public inputs ($k,\alpha$).

\begin{itemize}
\item State after the Server measures the second register:

\EQ{
\left(\ket{k,t_k,\alpha}\right)_C\otimes\left(\ket{k,\alpha}\otimes (\ket{x}+\ket{x'})\otimes\ket{y}\right)_S
}

\item State after the Server measures the first registers in $\alpha$ angles:

\EQ{
\left(\ket{k,t_k,\alpha,y}\right)_C\otimes\left(\ket{k,\alpha,y}\otimes \ket{b}\otimes\ket{Output}\right)_S
}
where $\ket{Output}=\ket{+_\theta}$ if there is no abort, while $\ket{Output}=\ket{x_n}$ otherwise.
\end{itemize}
The final state is

\EQ{\rho_f(\rho_{in})&=&
\left(\ket{k,t_k,\alpha,b}\right)_C\otimes\left(\ket{k,\alpha,y,b}\otimes\ket{+_\theta}\otimes\ket{\mathsf{no- abort}}\right)_S \ \textrm{if no abort}\nonumber\\
&=& \left(\ket{k,t_k,\alpha,b}\right)_C\otimes\left(\ket{k,\alpha,y,b}\otimes\ket{x_n}\otimes \ket{\mathsf{abort}}\right)_S \ \textrm{if abort}
}
To obtain the corresponding view, the Simulator calls the ideal functionality, but only uses the $\mathsf{abort}/\mathsf{no-abort}$ decision, and otherwise acts as in previous steps: Generates the state $\hat{\rho}_f$ (from the no-extra information lemma), obtains the final state $\rho_f(\rho_{in})$ by running the actual protocol until the previous step and adding the extra register $\ket{\mathsf{abort}}/\ket{\mathsf{no-abort}}$, and then applies the inverse of the isometry $T_f$ and traces out the Client's registers. Note that, as given in the definitions, all operators used correspond to polynomially-sized quantum circuits and therefore the Simulator is also QPT.
\end{proof}

Before moving to the constructions of trapdoor functions with the required properties and discussing the malicious case, we need to make an important observation. The ideal Protocol 
\ref{ideal:q_factory} other than the classical information $(k,y,\alpha,b)$, returns the state $\ket{+_\theta}$ to the Server. The security of our real \autoref{protocol:concrete_c_q_factory} that we proved is with respect to the ideal protocol (i.e. no information beyond that of the ideal protocol is obtained). However, having access to (a single copy) of the state $\ket{+_\theta}$ can (and does) give some non-negligible information on the classical description of that specific $\theta$. For example, by making a measurement one can rule-out one of the eight states with certainty. This, naively, would appear to be in contradiction with the properties of the function we have (where we prove that one can have only negligible advantage in guessing $\theta$). This is no different from the SRQG functionality, that the server can obtain some information on $r_m$. The resolution to this apparent contradiction, is that the basis of the proof of the hard-core property of $\theta$ with respect to the function,  is that one can repeat the same guessing algorithm keeping same $x$ (or $y$) but varying $\alpha$'s. However, to obtain any information from the (output) qubit, one needs to measure it and disturb it. Then repeating the experiment the probability of obtaining the same $y$ a second time (and thus having prepared the same $\theta$) is negligible for any QPT adversary (if one can repeat only polynomial number of times). Therefore, this one-shot extra information on $\theta$, cannot be distinguished from a one-shot information on a truly random $r_m$.

\section{Function Constructions}\label{Sec:Functions}

For our \autoref{protocol:concrete_c_q_factory} we need a trapdoor one-way function that is also quantum-safe, two-regular and second preimage resistant (or the stronger collision resistance property). These properties may appear to be too strong to achieve, however, we give here methods to construct functions that achieve these properties starting from trapdoor one-way functions that have fewer (more realistic) conditions, and we specifically give one example that achieves \emph{all} the desired properties. In particular we give:

\begin{itemize}
\item A general construction given either (i) an injective,
 homomorphic (with respect to any operation\footnote{in particular it is only required to be homomorphic once - rephrase this}) trapdoor one-way function or (ii) a bijective trapdoor one-way function, to obtain a two-regular, second preimage resistant\footnote{In (i) we prove the stronger collision-resistant property.},
trapdoor one-way function. In both cases the quantum-safe property is maintained (if the initial function has this property, so does the constructed function).
\item (taken from \cite{MP2012}) A method of how to realise injective quantum-safe trapdoor functions derived from the \LWE{} problem, that has certain homomorphic property. 
\item A way to use the first construction with the trapdoor from \cite{MP2012} that requires a number of modifications, including relaxation of the notion of two-regularity. The resulting function satisfy all the desired properties if a choice of parameters that satisfy multiple constraints, exists. 
\item A specific choice of these parameters, satisfying all constraints, that leads to a concrete function with all the desired properties.
\end{itemize}

\subsection{Obtaining two-regular, collision resistant/second preimage resistant, trapdoor one-way functions}\label{sec:general_two_regular}
Here we give two constructions. The first uses as starting point an injective, homomorphic trapdoor function while the second a bijective trapdoor function. While we give both constructions, we focus on the first construction since (i) we can prove the stronger collision-resistance property and (ii)
(to our knowledge) there is no known bijective trapdoor function that is believed to be quantum-safe.

\begin{theorem}\label{Thm:injective_to_desired}
If $\mathcal{G}$ is a family of injective, homomorphic, trapdoor one-way functions, then there exists a family $\mathcal{F}$ of two-regular, collision resistant, trapdoor  one-way functions. Moreover the family $\mathcal{F}$ is quantum-safe if and only if the family $\mathcal{G}$ is quantum-safe.
\end{theorem}

From now on, we consider that any function  $g_k \in \mathcal{G}$ has domain $D$ and range $R$ and let $\square$ be the closed operation on $D$ and $\star$ be the closed operation on $R$ such that $g_k$ is the morphism between $D$ and $R$ with respect to  these 2 operations: 
$$ g_k(a) \star g_k(b) = g_k(a \, \square \, b) \, \, \, \, \forall a,b \in D $$
We also denote the operation $\triangle$ on $D$, the inverse operation of $\square$, specifically: $ a \, \square \, b^{-1} = a \, \triangle \, b \, \, \, \, \forall a,b \in D $ and $0$ be the identity element for $\square$. \\
Then, the family $\mathcal{F}$ is described by the following PPT algorithms: \\

\procedure [linenumbering]{{\tt FromInj.Gen}$_{\mathcal{F}}(\secparam) $}{
	(k, t_k) \sample Gen_{\mathcal{G}}(\secparam) \, \, \,\, \, \, \pccomment{ $k$ is an index of a function from $\mathcal{G}$ and $t_k$ is its associated trapdoor} \\
	x_0 \sample D \setminus \{0\} \, \, \, \pccomment{$x_0 \neq 0$ to ensure that the 2 preimages mapped to the same output are distinct} \\
	k' := (k, g_k(x_0)) \, \, \, \pccomment{the description of the new function} \\
	t_k' := (t_k, x_0) \, \, \, \pccomment{ the trapdoor associated with the function $f_{k'}$} \\
	\pcreturn k', t_k' 
}\\

The Evaluation procedure receives as input an index $k'$ of a function from $\mathcal{F}$ and an element $\bar{x}$ from the function's domain ($\bar{x} \in D \times \{0, 1\}$): \\

\procedure {{\tt FromInj.Eval}$_{\mathcal{F}}(k', \bar{x})$}{
\pcreturn  f_{k'}(\bar{x})
}\\ 
where every function from $\mathcal{F}$ is defined as:
\begin{equation}
\boxed{
\begin{aligned}[b]
	& f_{k'} : D \times \{0, 1\} \rightarrow R \nonumber \\  
	& f_{k'}(x,c) = 
     \begin{cases}
      g_k(x)  \text{, } &\quad\text{if } c = 0\\
      g_k(x) \star g_k(x_0) = g_k(x \, \square \, x_0) \footnotemark\ \text{, } &\quad\text{if } c = 1\\ 
     \end{cases} \nonumber
\end{aligned} 
}
\end{equation}

\footnotetext{The last equality follows since each function $g_k$ from $\mathcal{G}$ is homomorphic}

\procedure [linenumbering]{{\tt FromInj.Inv}$_{\mathcal{F}}(k', y, t_k')$} {
	\pccomment{y is an element from the image of $f_{k'}$, $k' = (k, g_k(x_0)), \, \, t_k' = (t_k, x_0)$} \\
	x_1 := Inv_{\mathcal{G}}(k, y, t_{k}) \\
	x_2 := x_1 \triangle x_0 \\
	\pcreturn (x_1, 0) \, \, \, and \, \, \, (x_2, 1)  \, \, \, \pccomment{the unique 2 preimages corresponding to } \\
	\tab \tab \tab \tab \pccomment{an element from the image of $f_{k'}$}
}

\begin{proof}
To prove \autoref{Thm:injective_to_desired} we give below five lemmata showing that, the family $\mathcal{F}$ of functions defined above, satisfies the following properties: (i) two-regular, (ii) trapdoor, (iii) one-way, (iv) collision-resistant and (v) quantum-safe if $\mathcal{G}$ is quantum-safe. 
\end{proof}

\begin{lemma}[two-regular]\label{lemma:two-regular}
If $\mathcal{G}$ is a family of injective, homomorphic functions, then $\mathcal{F}$ is a family of two-regular functions.
\end{lemma}

\begin{proof}

For every $y \in \Im f_{k'} \subseteq R$, where $k'=(k,g_k(x_0))$:
\begin{enumerate}
\item Since $\Im f_{k'}= \Im g_k$ and $g_k$ is injective, there exists a unique $x:=g^{-1}_k(y)$ such that $f_{k'}(x,0)=g_k(x)=y$.
\item Assume $x'$ such that $f_{k'}(x',1)=y$. By definition $f_{k'}(x',1)=g_k(x' \, \square \, x_0) = y$, but $g_k$ is injective and $g_k(x) = y$ by assumption, therefore there exists a unique $x'= x \triangle x_0$ such that $f_{k'}(x',1) = y$
\end{enumerate}
Therefore, we conclude that:

\EQ{\label{eq:preimages}
\forall \ y \ \in \Im f_{k'}: \ f_{k'}^{-1}(y):=\{(g_k^{-1}(y),0),(g_k^{-1}(y) \, \triangle \, x_0,1)\} 
} 

\end{proof}

\begin{lemma}[trapdoor]
  If $\mathcal{G}$ is a family of injective, homomorphic, trapdoor functions, then $\mathcal{F}$ is a family of trapdoor functions.
\end{lemma}

\begin{proof}
Let $y \in \Im f_{k'} \subseteq R$. We construct the following inversion algorithm:\\

\procedure [linenumbering]{$Inv_{\mathcal{F}}(k', y, t_k')$} {
\pccomment{ $t_k' = (t_k, x_0)$,  $k' = (k, g_k(x_0))$}\\
x := Inv_{\mathcal{G}}(k, y, t_k)\\
\pcreturn (x,0) \text{ and } (x \, \triangle \, x_0, 1)
}

\end{proof}

\begin{lemma}[one-way]\label{lemma:onewayFromInj}
If $\mathcal{G}$ is a family of injective, homomorphic, one-way functions, then $\mathcal{F}$ is a family of one-way functions.
\end{lemma}

\begin{proof}
We prove it by contradiction. We assume that a QPT adversary $\mathcal{A}$ can invert any function in $\mathcal{F}$ with non-negligible probability $P$ (i.e. given $y\in \Im f_{k'}$ to return a correct preimage of the form $(x',b)$ with probability $P$). We then construct a QPT adversary $\mathcal{A'}$ that inverts a function in $\mathcal{G}$ with the same non-negligible probability $P$ reaching the contradiction since $\mathcal{G}$ is one-way by assumption.

From Eq. (\ref{eq:preimages}) of \autoref{lemma:two-regular} we know the two preimages of $y$ are: (i) $(g_k^{-1}(y),0)$ and  (ii) $(g^{-1}_k(y) \triangle x_0,1)$. We see that information on $g^{-1}_k(y)$ is obtain in both cases, i.e. obtaining any of these two preimages, is sufficient to recover $g_k^{-1}(y)$ if $x_0$ is known. We now construct an adversary $\mathcal{A'}$ that for any function $g_k : D \rightarrow R$, inverts any output $y=g_k(x)$ with the same probability $P$ that $\mathcal{A}$ succeeds.\\

\procedure [linenumbering] {$\mathcal{A}'(k, y)$} {
	x_0 \sample D \setminus \{0\} \pccomment{$\mathcal{A'}$ knows $x_0$, but is not given to $\mathcal{A}$}\\
	k' := (k, g_k(x_0)) \\
	(x', b) \gets \mathcal{A}(k', y) \\
	\pcif ((b == 0) \wedge (g_k(x') = y) \pcthen\\
        \t \pccomment{equivalent to $\mathcal{A}$ succeeded in returning the first preimage}	 \\
	\t \pcreturn x' \\
	\pcelseif ((b == 1) \wedge (g_k(x' \, \square \, x_0)) = y) \pcthen\\
        \t \pccomment{$\mathcal{A}$ succeeded in returning the second preimage} \\
	\t \pcreturn  x' \, \square \, x_0\pccomment{$\mathcal{A'}$ uses $x_0$ known from step 1}  \\
	\pcelse \pccomment{$\mathcal{A}$ failed in giving any of the preimages (happens with probability $1-P$)}\\
	\t \pcreturn 0 
}\\

\end{proof}

\begin{lemma}[collision-resistance]
If $\mathcal{G}$ is a family of injective, homomorphic,  one-way functions, then any function $f \in \mathcal{F}$ is collision resistant.
\end{lemma}

\begin{proof}
Assume there exists a QPT adversary $\mathcal{A}$ that given $k'=(k,g_k(x_0))$ can find a collision $(y, (x_1,b_1),(x_2,b_2))$ where $f_k(x_1,b_1)=f_{k'}(x_2,b_2)=y$ with non-negligible probability $P$. From Eq. (\ref{eq:preimages}) we know that the two preimages are of the form $(x,0),(x \, \triangle \, x_0,1)$ where $g_k(x)=y$. It follows that when $\mathcal{A}$ is successful, by comparing the first arguments of the two preimages, can recover $x_0$. 

We now construct a QPT adversary $\mathcal{A'}$ that inverts the function $g_k$ with the same probability $P$, reaching a contradiction:\\

\procedure [linenumbering]{$\mathcal{A'}(k, g_k(x))$} {
k' = (k, g_k(x)) \\
(y, (x_1, b_1), (x_2, b_2)) \wedge x_1\neq x_2 \gets \mathcal{A}(k') \pccomment{where $y$ is an element from the image of $f_{k'}$} \\
\pcif f(x_1,b_1)=f(x_2,b_2)=y\\
\pcreturn x= x_1 \, \triangle \, x_2 \\
\pcelse \pccomment{$\mathcal{A}$ failed to find collision of $f_{k'}$; happens with probability $(1-P)$}\\
\pcreturn 0
}

\end{proof}

\begin{lemma}[quantum-safe]
If $\mathcal{G}$ is a family of quantum-safe trapdoor functions, with properties as above, then $\mathcal{F}$ is also a family of quantum-safe trapdoor functions.
\end{lemma}

\begin{proof}
The properties that require to be quantum-safe is the one-wayness and collision resistance. Both these properties of $\mathcal{F}$ that we derived above were proved using reduction to the hardness (one-wayness) of $\mathcal{G}$. Therefore if $\mathcal{G}$ is quantum-safe, its one-wayness is also quantum-safe and thus both properties of $\mathcal{F}$ are also quantum-safe.
\end{proof}

\begin{theorem}\label{Thm:bijective_to_desired}
If $\mathcal{G}$ is a family of bijective, trapdoor one-way functions, then there exists a family $\mathcal{F}$ of two-regular, second preimage resistant, trapdoor one-way functions. Moreover ,the family $\mathcal{F}$ is quantum-safe if and only if the family $\mathcal{G}$ is quantum-safe.
\end{theorem}
The family $\mathcal{F}$ is described by the following PPT algorithms, where each function $g_k \in \mathcal{G}$ has domain $D$ and range $R$:\\

\procedure [linenumbering] {{\tt FromBij.Gen}$_{\mathcal{F}}(\secparam) $} {
	(k_1, t_{k_1}) \sample Gen_{\mathcal{G}}(\secparam) \\
	(k_2, t_{k_2}) \sample Gen_{\mathcal{G}}(\secparam) \\
	k' := (k_1, k_2) \\
	t_k' := (t_{k_1}, t_{k_2}) \\
	\pcreturn k', t_k' 
}\\ \\

\procedure {{\tt FromBij.Eval}$_{\mathcal{F}}(k', \bar{x})$}{
\pcreturn  f_{k'}(\bar{x})
}\\ \\
where every function from $\mathcal{F}$ is defined as:

\begin{equation}
\boxed{
\begin{aligned}[b]
	& f_{k'} : D \times \{0, 1\} \rightarrow R \nonumber \\
	& f_{k'}(x,c) = 
     \begin{cases}
      g_{k_1}(x)  \text{, } &\quad\text{if } c = 0\\
      g_{k_2}(x)  \text{, } &\quad\text{if } c = 1\\ 
     \end{cases}  \nonumber
\end{aligned} 
}
\end{equation}

\procedure [linenumbering]{{\tt FromBij.Inv}$_{\mathcal{F}}(k', y, t_k')$} {
	\pccomment{y is an element from the image of $f_{k'}$, $k' = (k_1, k_2), \, \, t_k' = (t_{k_1}, t_{k_2})$} \\
	x_1 := Inv_{\mathcal{G}}(k_1, y, t_{k_1}) \\
	x_2 := Inv_{\mathcal{G}}(k_2, y, t_{k_2}) \\
	\pcreturn (x_1, 0) \, \, \, and \, \, \, (x_2, 1)  \, \, \, \pccomment{the unique 2 preimages corresponding to } \\
	\tab \tab \tab \tab \pccomment{an element from the image of $f_{k'}$}
}

The proof of \autoref{Thm:bijective_to_desired}, using the family of function defined above, follows same steps as of \autoref{Thm:injective_to_desired} and is given in the \autoref{app:bijective}.

\subsection{Injective, homomorphic quantum-safe trapdoor one-way function from \LWE{} (taken from \cite{MP2012})}

We outline the Micciancio and Peikert \cite{MP2012} construction of injective trapdoor one-way functions, naturally derived from the Learning With Errors problem. At the end we comment on the homomorphic property of the function, since this is crucial in order to use this function as the basis to obtain our desired two-regular, collision resistant trapdoor one-way functions.

The algorithm below generates the index of an injective function and its corresponding trapdoor. The matrix $G$ used in this procedure, is a fixed matrix (whose exact form can be seen in \cite{MP2012}) for which the function from the family $\mathcal{G}$ with index $G$ can be efficiently inverted.

\procedure [linenumbering]{{\tt LWE.Gen}$_{\mathcal{G}}(1^n) $}{
  A' \sample \mathbb{Z}_q^{n \times \bar{m}} \\
  \pccomment{$\cD$ denotes the element-wise gaussian distribution with mean 0}\\
  \pccomment{and standard deviation $\alpha q$ on matrices of size $\bar{n} \times kn$}\\
  R \sample   \mathcal{D}^{\bar{m} \times kn}_{\alpha q} \, \, \, \pccomment{trapdoor information} \\
  A := (A', G - A'R) \, \, \, \pccomment{concatenation of matrices A' and G - A'R, representing the index of the function} \\
  \pcreturn A, R
} \\ 

The actual description of the injective trapdoor function is given in the Evaluation algorithm below, where each function from $\mathcal{G}$ is defined on: $ g_K : \mathbb{Z}_q^n \times L^m \rightarrow \mathbb{Z}_q^m$, and $L$ is the domain of the errors in the \LWE{} problem (the set of integers bounded in absolute value by $\mu$):

\procedure [linenumbering]{{\tt LWE.Eval}$_{\mathcal{G}}(K, (s, e)) $}{
	y := g_K(s, e) = s^tK + e^t \\
    \pcreturn y
}\\

The inversion algorithm returns the unique preimage $(s, e)$ corresponding to $b^t \in \Im (g_K)$. The algorithm uses as a subroutine the efficient algorithm $Inv_G$ for inverting the function $g_G$, with $G$ the fixed matrix mentioned before.

\procedure {{\tt LWE.Inv}$_{\mathcal{G}}(K, t_K, b^t) $}{
  \pcln {b'}^t := b^t \begin{bmatrix} R \tabularnewline I \end{bmatrix} \\
  \pcln (s', e') := Inv_G(b') \\
  \pcln s := s' \\
  \pcln e := b - K^ts \\
  \pcln \pcreturn s, e
}\\

We examine now whether the functions $g_K$ are homomorphic with respect to some operation. 
Given $a = (s_1, e_1) \in  \mathbb{Z}_q^n \times L^m $ and $b = (s_2, e_2) \in  \mathbb{Z}_q^n \times L^m$, the operation $\star$ is defined as:
$$ (s_1, e_1) \, \star \, (s_2, e_2)  = (s_1 + s_2 \bmod q, e_1 + e_2)$$
Given $y_1 = g_K(a) \in \mathbb{Z}_q^m$ and $y_2 = g_K(b) \in \mathbb{Z}_q^m$, the operation $\square$ is defined as:
$$ y_1 \, \square \, y_2 = y_1 + y_2 \bmod q $$
Then, we can easily verify that:
\EQ{& g_K(s_1, e_1) + g_K(s_2, e_2) \bmod q = {s_1}^tK + {e_1}^t + {s_2}^tK + {e_2}^t \bmod q = \nonumber\\
& (s_1 + s_2 \bmod q)^tK + (e_1 + e_2)^t = g_K((s_1 + s_2) \bmod q, e_1 + e_2)\nonumber}
However, the sum of two error terms, each being bounded by $\mu$, may not be bounded by $\mu$. This means that the function is not (properly) homomorphic. Instead, what we conclude is that as long the vector $e_1 + e_2$ lies inside the domain of $g_K$, then $g_K$ is homomorphic. To address this issue, we will need to define a weaker notion of 2-regularity, and a (slight) modification of the {\tt FromInj} construction to provide a desired function starting from the trapdoor function of \cite{MP2012}.

\subsection{A suitable $\delta$-2 regular trapdoor function}\label{Subsec:actual_trapdoor}

Using the homomorphic injective trapdoor function of Micciancio and Peikert \cite{MP2012} and the construction defined in the proof of \autoref{Thm:injective_to_desired}, we derive a family $\mathcal{F}$ of collision-resistant trapdoor one-way function, but with a weaker notion of 2-regularity, called $\delta$-2 regularity:
\begin{definition}[$\delta$-2 regular]
A family of functions $(f_i)_{i \leftarrow Gen_{\mathcal{F}}}$ is said to be \mbox{$\delta$-2 regular}, with $\delta \in [0,1]$ if:
  \[\Pr_{i \leftarrow Gen_{\mathcal{F}}, y \in \Im (f_i)} [~ |f^{-1}_i(y)| = 2 ~] \geq \delta \]
\end{definition}

Given this definition, we should note here that in \autoref{protocol:concrete_c_q_factory} we need to modify the abort case to include the possibility that the image $y$ obtained from the measurement does not have two preimages (something that happens with at most probability $(1-\delta)$).

\begin{theorem}[Existence of a $\delta$-2 regular trapdoor function family]
  There exists a family of functions that are $\delta$-2 regular (with $\delta$ at least as big as a fixed constant), trapdoor, one-way, collision resistant and quantum-safe, assuming that there is no quantum algorithm that can efficiently solve $\SIVP{}_\gamma$ for $\gamma = \poly[n]$.
\end{theorem}
\begin{proof}
To prove this theorem, we define a function similar to the one in the {\tt FromInj} construction, where the starting point is the function defined in \cite{MP2012}. Crucial for the security is a choice of parameters that satisfy a number of conditions given by \autoref{thm:req_param} and proven in \autoref{app:implementation}. The proof is then completed by providing a choice of parameters given in \autoref{thm:exist_param} that satisfies all conditions as it is shown in \autoref{app:implementation2}.  
\end{proof}

\begin{definition}\label{def:REG2_fct}
  For a given set of parameter $\cP$ chosen as in \autoref{thm:req_param}, we define the following functions, that are similar to the construction {\tt FromInj}, except for the key generation that require an error sampled from a smaller set:\\
  \begin{minipage}[t]{.45\textwidth}
    \procedure [linenumbering]{{\tt REG2.Gen}$_\cP(1^n) $}{
      A, R \gets {\tt LWE.Gen}_{\cG}(1^n)\\
      s_0  \gets \Z_q^{n,1} \\
      e_0 \gets \cD^{m,1}_{\alpha' q}\\
      b_0 := {\tt LWE.Eval}(A, (s_0 + e_0))\\
      k := (A,b_0)\\
      t_k := (R,(s_0,e_0))\\
      \pcreturn (k, t_k)
    }
  \end{minipage}
  \begin{minipage}[t]{.45\textwidth}
    \procedure [linenumbering]{{\tt REG2.Eval}$_\cP((A,b_0), (s,e,c)) $}{
      \pccomment{$s$ is a random element in $\Z_q^{n,1}$, $c \in \{0, 1\}$}\\
      \pccomment{$e$ is sampled uniformly and such that}\\
      \pccomment{\t each component is smaller than $\mu$}\\
      \pcreturn {\tt LWE.Eval}(A, (s, e)) + c b_0
    } \\
    
    \procedure [linenumbering]{{\tt REG2.Inv}$_\cP((A, R,(s_0,e_0)), b) $}{
      (s_1,e_1) := {\tt LWE.Inv}(R, b)\\
      \pcif ||e_1 - e_0||_\infty \leq \mu \pcthen \  \pcreturn \bot \\
      \pcreturn ((s_1,e_1,0), (s_1-s_0,e_1-e_0,1))
    } \\
  \end{minipage}
\end{definition}
Note, that the pairs $(s,e)$ and $(s_0,e_0)$ correspond to $x$ and $x_0$ of the {\tt FromInj} construction of \autoref{sec:general_two_regular}. The idea behind this construction is that the noise of the trapdoor is sampled from a set which is small compared to the noise of the input function. That way, when you will add the trapdoor together with an input, the total noise will still be small enough to stay in the set of possible input noise with good probability, mimicking the homomorphic property needed in \autoref{Thm:injective_to_desired}. Note that the parameters need to be carefully chosen, and a trade-off between probability of success and security exists.

\begin{lemma}[Requirements on the parameters]\label{thm:req_param}
  For all $n,q, \mu \in \Z, \mu' \in \R$, let us define:
  \begin{itemize}
  \item $k := \ceil{log(q)}$
  \item $\bar{m} = 2n$
  \item $\omega = nk$
  \item $m := \bar{m} + \omega = 2n + nk$
  \item $\alpha' = \frac{\mu'}{\sqrt{m}q} $
  \item $\alpha  = m \alpha' $
  \item $C$ the constant in Lemma 2.9 of \cite{MP2012} which is around $\frac{1}{\sqrt{2\pi}}$
  \item $B = 2$ if $q$ is a power of 2, and $B = \sqrt{5}$ otherwise.
  \end{itemize}
  Now, if for all security parameters $n$ (dimension of the lattice), there exist $q$ (the modulus of \LWE{}) and $\mu$ (the maximum amplitude of the components of the errors) such that:
  \begin{enumerate}
  \item $m$ is such that $n = o(m)$ (required for the injectivity of the function (see e.g. \cite{Lecture13})) 
  \item $0 < \alpha < 1$
  \item $\mu' = O(\mu / m)$ (required to have non negligible probability to have two preimages)
  \item $\alpha' q \geq 2 \sqrt{n}$ (required for the \LWE{} to \SIVP{} reduction)
  \item $\frac{n}{\alpha'}$ is \poly[n] (representing, up to a constant factor, the approximation factor $\gamma$ in the $\SIVP{}_\gamma$ problem)
  \item \[\sqrt{m} \mu < \underbrace{\frac{q}{2 B \sqrt{\left(C \cdot (\alpha \cdot q) \cdot (\sqrt{2n} + \sqrt{kn} + \sqrt{n})\right)^2+1}}}_{r_{max}} - \mu'\sqrt{m} \]
    (required for the correctness of the inversion algorithm - $r_{max}$ represents the maximum length of an error vector that one can correct using the \cite{MP2012} function\footnote{We chose to use the computational definition of \cite{MP2012}, but this theorem can be easily extended to other definitions of this same paper, or even to other construction of trapdoor short basis)}, and the last term is needed in the proof of collision resistance to ensure injectivity even when we add the secret trapdoor noise, as illustrated in \autoref{fig:picture_balls_squares})
  \end{enumerate}
  then the family of functions of \autoref{def:REG2_fct} is $\delta$-2 regular (with $\delta$ at least as big as a fixed constant), trapdoor, one-way and collision resistant (all these properties are true even against a quantum attacker), assuming that there is no quantum algorithm that can efficiently solve $SIVP{}_\gamma$ for $\gamma = \poly[n]$.
\end{lemma}
\begin{proof}
The proof follows by showing that the function with these constraints on the parameters is (i) $\delta$-2 regular, (ii) collision resistant, (iii) one-way and (iv) trapdoor. In \autoref{app:implementation} we give and prove one lemma for each of those properties. For an intuition of the choice of parameters see also \autoref{fig:picture_balls_squares}. 
\end{proof}

% \begin{figure}
%   \begin{minipage}[c]{0.6\textwidth}
%     \includegraphics[width=\textwidth]{picture_balls_squares.pdf}
%   \end{minipage}\hfill
%   \begin{minipage}[c]{0.4\textwidth}
%     \caption{The red circle represents the domain of the error term from the trapdoor information, which is being sampled from a Gaussian distribution. The orange square is an approximation of this domain, which must satisfy that its length is much smaller (by a factor of at least $m$ = the dimension of the error) than the length of the blue square, used for the actual sampling from the domain of the error terms, for which it is known that the trapdoor function is invertible, domain represented by the green circle. The dashed part is needed to ensure that if there is a collision $x,x'$, then $x = x' \pm x_0$.  \label{fig:picture_balls_squares}}
%   \end{minipage}
% \end{figure}

%\begin{figure}
%  \centering
%  \begin{subfigure}{.48\textwidth}
%    \centering
%    \includegraphics[width=\textwidth]{picture_balls_squares.pdf}
\begin{figure}[h]
  \begin{minipage}[c]{0.5\textwidth}
    \includegraphics[width=\textwidth]{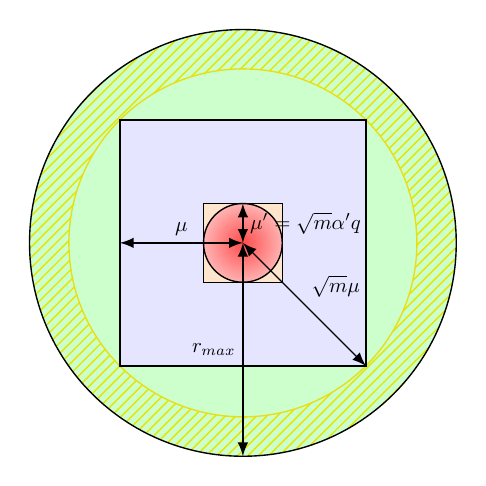}
  \end{minipage}\hfill
  \begin{minipage}[c]{0.4\textwidth}
    % 
    % \caption{Representation of the noise parameters}
    % \label{fig:sub_pict_balls1}
    % \end{subfigure}\hfill%
    % \begin{subfigure}{.48\textwidth}
    %   \centering
    %   \includegraphics[width=\textwidth]{tricky_abort.pdf}
    %   \caption{If we do not add a cancel part, then we can have unexpected collisions (here $e_2' \neq e_1-e_0$) in extreme cases.}
    %   \label{fig:sub_pict_balls1}
    % \end{subfigure}
    \caption{The red circle represents the domain of the error term from the trapdoor information, which is being sampled from a Gaussian distribution. The orange square is an approximation of this domain, which must satisfy that its length is much smaller (by a factor of at least $m$ -- the dimension of the error) than the length of the blue square, used for the actual sampling from the domain of the error terms, for which it is known that the trapdoor function is invertible, domain represented by the green circle. The dashed part is needed to ensure that if there is a collision $(x_1, x_2)$, then $x_1 = x_2 \pm x_0$.  \label{fig:picture_balls_squares}}
  \end{minipage}
\end{figure}

\newpage%
\subsection{Parameter Choices}

\begin{lemma}[Existence of parameters]\label{thm:exist_param}
  The following set of parameters fulfills \autoref{thm:req_param}.
  \begin{align*}
   n &= \lambda\\
   k &= 5\ceil{log(n)} + 21\\
   q &= 2^k\\
   \bar{m} &= 2n\\
   \omega &= nk\\
   m &= \bar{m} + \omega\\
   \mu &= \ceil{2mn \sqrt{2+k}}\\
   \mu' &= \mu/m\\
   B &= 2\\
 \end{align*}
 and $\alpha, \alpha', C$ are defined like in \autoref{thm:req_param}. 
\end{lemma}
The proof is given in \autoref{app:implementation2}. As a final remark, we stress that other choices of the parameters are possible (considering the trade-off between security and probability of success) and we have not attempted to find an optimal set.

%there are several different usable parameters, depending on the choice of the distribution $\mathcal{D}$ (we chose the computational distribution presented in \cite{MP2012}, but you can find other distributions in this paper), and also depending on the ratio between the small hypercube that contains the secret Gaussian sampling and the big input hypercube. When you decrease the size of the secret, you increase the probability of success, but you decrease the security, so all kind of trade-off are possible.

\section{Discussion}\label{Sec:Malicious}

In this work we deal with Quantum-Honest-But-Curious adversaries. Naturally, the final aim should be to provide security against a (fully) malicious adversary/Server. There are two (linked) issues to consider when dealing with malicious adversaries. The \textbf{first issue} is whether the Server (by deviating arbitrarily) can obtain extra information about the secret classical description (of the state supposingly prepared). The \textbf{second issue} is whether the actual state at the end of the protocol is (essentially) the one that the Client believes, i.e. if the functionality provides verification. We make remarks separately on these issues, and then conclude with an approach that could lead in a solution to both issues.

\noindent\textbf{Issue 1 (privacy):} The most naive attempt for the Server to deviate in order to obtain information, is to return $y,b$ other than those obtained from an honest run of the protocol. Since $y,b$ \emph{determine} (along with other parameters) the value of the secret $\theta$ a deviation there could lead to breaking the security. For example, instead of the (truly) random $y$ that is obtained in the honest run, the Server can choose $y$ such that he has information on the preimages for the given $k$ or can choose $b$ adaptively depending on values of $\alpha$. 
However, the function $f_k$ is \emph{collision}-resistant, which means that even if the adversary chooses the $y$ he cannot find such $y$ that both preimages are known with a non-negligible probability. Moreover, if the Server chooses $y$, it means that the protocol was not followed and thus the final output state is not going to be related with the value $\theta$ as expected. We conjecture the hard-core function proof (\autoref{Thm:hardcore}) will remain valid in that case with minor modifications. The more significant difficulty, however, comes from ``mixed'' strategies, where the adversary follows partly the protocol (and thus the output qubit is correlated with the classical secret description), and partly deviates. In those cases it is hard to quantify what information the server has, and whether this is strictly less than that of an ideal protocol (where the state $\ket{+_\theta}$ gives some legitimate information).

\noindent\textbf{Issue 2 (verification):} The first thing to note, is that the adversary has in his lab the output state, and therefore (trivially) he can always apply a final-step deviation corrupting the legitimate output. Thus when we speak of verification, we mean a correct state, up to a (reversible) deviation on the Server's side (as the operations $\mathcal{T}_i$ in the definition of specious). The second thing to stress, is that \autoref{protocol:concrete_c_q_factory} cannot be verifiable against a malicious Server, unless some extra mechanism is added. There is a way, by deviating from the instructions, to corrupt the output in a way that depends on the secret classical description ($\theta$), but without actually learning any information about the same classical description. In particular it is possible by measuring all qubits of the first register in $3\alpha$ angle to generate the state $\ket{+_{3\theta}}$ as output. This deviation does not help the Server to learn \emph{any} information about $\theta$ (protocol still ``private'') but affects the output state in a ``non-reversible'' way and thus compromises the verifiability.

\noindent \textbf{A way forward:} The ultimate goal would be to extent QFactory into a Quantum Universal Composable protocol \cite{unruh2010uni} in order to be able to compose it with any other protocol, or at least to proof the security against a malicious adversary. In classical protocols (and recently in quantum too \cite{KMW2017}), the way to boost the security from honest-but-curious to malicious is to introduce a ``compiler'' (e.g. using the construction in \cite{GMW87} or a cut-and-choose technique) and boost the security by essentially enforcing the honest-but-curious behaviour to malicious adversaries (or abort). In our case, the protocol is simple enough, having single qubits as outputs. One method could be to prepare a large string of qubits, and have the Client choose a random subset of those and instructs the Server to measure them. By observing the correct statistics on the ``test'' qubits, one can infer the correct preparation. This is closely related to the parameter-estimation in QKD, and with self-testing \cite{MYS2012}. The exact details are involved, as the analogous cases of compilers, parameter-estimation and self-testing suggests, and will be explored in a future publication.

\section*{Acknowledgements}

A.C. and P.W. are very grateful to Thomas Zacharias, Aggelos Kiayias and especially Yiannis Tselekounis for many useful discussions about the security proofs. L.C. also thanks Atul Mantri and Dominique Unruh for useful preliminary conversations, with a special mention to C\'{e}line Chevalier whose discussions were a precious source of inspiration.
A.C. would also like to show his appreciation to his grandmother, Petra Ilie for all the support and help she has given to him his entire life.
The work was supported by EPSRC grants EP/N003829/1 and EP/M013243/1.

\newpage%
\appendix

\begin{appendices}

\section{PSRQG within several applications}
\label{app:applications}

In \autoref{Sec:applications} we listed several applications that can use the PSRQG functionality to allow for fully classical parties to participate using, a potentially malicious, quantum server. Here we give details on how to use the exact output of our QFactory protocol in these applications. We emphasize that in all protocols in which the ``server'' used by the classical party is a malicious party, the cost of using our QFactory construction is that the security becomes computational and applies in the quantum-honest-but-curious setting.
\begin{enumerate}
\item In the quantum homomorphic encryption scheme AUX in \cite{broadbent2015quantum}, where the target quantum computation must have constant $T$-gate depth, using our QFactory protocol would allow a classical client to participate (delegate such computation) provided, of course, that the input/output are classical.
Specifically, as the input is classical, the client will instruct the server to prepare a quantum state of the classical one-time pad of this input (and then the client will also send to the server a classical homomorphic encryption of the classical one-time pad key of each of the input's bits). Moreover, for every $T$-gate in the quantum computation, the auxiliary qubits in the evaluation key can be produced using QFactory: $\left\{\ket{+}, P\ket{+} = \ket{+_{2{\pi}/4}}, Z\ket{+} = \ket{+_{4\pi/4}}, ZP\ket{+} = \ket{+_{6\pi/4}}\right\}$. \\ 
We note that due to the use of a classical fully homomorphic encryption scheme, the AUX protocol \cite{broadbent2015quantum} has computationally security, thus, the computational security offered by the QFactory is not downgrading the security of this protocol. 

\item In the blind delegated quantum computation protocol of \cite{bfk}, the client needs to prepare and send to the server qubits, randomly chosen, from the set of states $\{\ket{+}, \ket{+_{\pi/4}}, ... , \ket{+_{7\pi/4}} \}$. This is exactly the set of states of Eq. (\ref{eq:output_states}) which are given by the QFactory. It follows that our construction eliminates the need for quantum communication and thus any classical client can use this protocol.

\item The verifiable blind quantum computation protocol in \cite{FKD2017}, the only quantum ability that the verifier needs is to prepare and send to the prover single qubits, randomly chosen, from the set of states $\{\ket{+_{k\pi/4}}\}$. Again, this is exactly the set of states given by the QFactory. Therefore, the quantum communication, and thus quantum abilities of the verifier, can be completely replaced by the QFactory functionality.

\item For the quantum key-distribution construction in \cite{BB84}, we can use two conjugate bases to realise this protocol, namely: the diagonal basis $\{\ket{+}, \ket{+_{\pi}}\}$ and the left-right handed circular basis $\{\ket{+_{\pi/2}}, \ket{+_{3\pi/2}}\}$. All these four quantum states can be obtained by the QFactory protocol\footnote{Actually, if one was interested in obtaining exactly and only this set of states, we can modify the QFactory to do so, in a way that actually simplifies the proofs too. For example, we could simply ask to measure the qubits in the second stage in the basis $\{\ket{\pm_{\alpha_i\pi/2}}\}$.}. As the quantum coin flipping protocol of \cite{BB84}, the quantum money protocol of \cite{BOV2018}, or the quantum digital signatures protocol of \cite{WDKA2015} only require, as in \cite{BB84}, any pair of conjugate bases, this implies that we can use QFactory in a straight forward way. On the other hand, for the quantum coin flip construction in \cite{PCDK2011}, the single qubit quantum states needed are of the form $ \sqrt{a}\ket{0} + {(-1)}^{{\alpha}_i}\sqrt{1 - a}\ket{1} $, which might be achieved by a different construction of the PSRQG.

\item In the multiparty quantum computation protocol of \cite{KP16}, the $n$ clients need to send multiple copies of quantum states in the set $\{\ket{+_{k\pi/4}}\}$ to the server, who entangles and measures them all but one. Using QFactory all these states will be prepared by the server, which would enable the $n$ clients to be fully classical.

\item The verifiable blind quantum computation protocols in \cite{Broadbent2015}, \cite{fk} or the two-party quantum computation protocols in \cite{KW2017}, \cite{KMW2017}, require the honest party to prepare single qubit states from the set of states $\{\ket{0}, \ket{1},  \ket{+_{k\pi/4}} \}$. While the QFactory primitive can output the $\ket{+_{k\pi/4}}$ states, in order to make the honest party fully classical, we need to change the construction of QFactory in order to also be able to output the $\ket{0}$ and $\ket{1}$ states, and maintain the same guarantees in privacy as in the QFactory.

\end{enumerate}

\section{Full proof of \autoref{Thm:hardcore}}
\label{app:hardcore}

\begin{proof}
  From Eq. (\ref{eq:proof1}) we have the definition of $\widetilde{B}$
  in terms of the three corresponding bits and we aim to prove that it
  is hard-core, i.e. that Eq. (\ref{eq:proof2}) is satisfied. We will
  follow the five steps outlined in the main text. Before that let us
  define some simple identities that will be used.
  $ \forall a,b,d,e \in \mathbb{N}$, we have:

\begin{equation}
\tag{I1}\label{I1} (a + b) \bmod 8 = (a \bmod 8 + b \bmod 8) \bmod 8 
\end{equation} 
\begin{equation}
\tag{I2}\label{I2} [(a + b) \bmod 8] \bmod 4 =  (a \bmod 4 + b \bmod 4) \bmod 4
\end{equation}
\begin{equation}
\tag{I3}\label{I3} [(a + b) \bmod 4] \bmod 2 =  (a \bmod 2 + b \bmod 2) \bmod 2
\end{equation}
\begin{equation}
\tag{I4}\label{I4} (2a) \bmod 4 = 2 \cdot (a \bmod 2)
\end{equation}
\begin{equation}
\tag{I5}\label{I5} (2a) \bmod 8 = 2 \cdot (a \bmod 4)
\end{equation}
\begin{equation}
\tag{I6}\label{I6} (2d + e) \bmod 4 - e \bmod 2 = [2d + e - (e \bmod 2)] \bmod 4
\end{equation}
\begin{equation}
\tag{I7}\label{I7} (2d + e) \bmod 8 - e \bmod 2 = [2d + e - (e \bmod 2)] \bmod 8
\end{equation}
We now return to Eq. (\ref{eq:proof1})

\[\widetilde{B} = g(x-x')=\sum \limits_{i = 1}^n (x_i - x_i')(4b_i + \alpha_i) \bmod 8\]

where $ \widetilde{B} = \widetilde{B_1}\widetilde{B_2}\widetilde{B_3} $, with $\widetilde{B_j} \in \{0, 1\}.$ We also define $\widetilde{x} = x \oplus x' \in \{0, 1\}^n$ and $z \in \{-1, 0, 1\}^n$ be the vector defined as: $z_i = x_i - {x'}_i = (-1)^{{x'}_i}\widetilde{x}_i  \, , \, \forall i \in \{1,2,...,n\}$. \\

\noindent\textbf{Step 1:} We will rewrite this expression in terms of
single bits and obtain the expression of Eq. (\ref{eq:proof2}).
We have $g(z) = \sum \limits_{i = 1}^n z_i(4b_i + \alpha_i) \bmod 8$, or equivalently: 
\EQ{\label{eq:app1} 4\widetilde{B_1} + 2\widetilde{B_2} + \widetilde{B_3} = \left[ \left(4 \sum \limits_{i = 1}^n z_ib_i \right) \bmod 8 + \left( \sum \limits_{i = 1}^n z_i \alpha_i \right) \bmod 8 \right] \bmod 8\nonumber
}       

\noindent We define the following terms:
$\alpha_i = 4{\alpha}_i^{(1)} + 2{\alpha}_i^{(2)} + {\alpha}_i^{(3)}$,
where ${\alpha}_i^{(1)}, {\alpha}_i^{(2)}, {\alpha}_i^{(3)}$ are the 3
bits of $\alpha_i$ and ${\alpha}^{(j)} \in \{0, 1\}^n$ are the vectors
consisting of the $j$-th bit of all values ${\alpha}_i$,
$\forall i \in \{1,2,...,n\}, j \in \{1,2,3\}$;

\EQ{S_0 = \sum \limits_{i = 1}^n z_i b_i & ; & S_1 = \sum \limits_{i = 1}^n z_i {\alpha}_i^{(1)}\nonumber\\
S_2 = \sum \limits_{i = 1}^n z_i {\alpha}_i^{(2)} & ; & S_3 = \sum \limits_{i = 1}^n z_i {\alpha}_i^{(3)}\nonumber
}
We also notice that under $\bmod \, 2$, we have that: 
\EQ{S_j \bmod 2 = \sum \limits_{i = 1}^n \widetilde{x_i} {\alpha}_i^{(j)} \bmod 2 = \langle \widetilde{x}, {\alpha}^{(j)} \rangle \bmod 2 \textrm{ , for }j \in \{1, 2, 3\}.\nonumber
} 
Then, we have: 

\EQ{\label{eq:app2}& 4\widetilde{B_1} + 2\widetilde{B_2} + \widetilde{B_3} = ( \sum \limits_{i = 
1}^n (x_i - x_i')(4b_i + \alpha_i)) \bmod 8 = \nonumber \\
& \left[ 4S_0 + \sum \limits_{i = 1}^n (x_i - x_i')( 4{\alpha}_1^{(i)} + 2{\alpha}_2^{(i)} + 
{\alpha}_3^{(i)}) \right] \bmod 8 \nonumber \\
& 4\widetilde{B_1} + 2\widetilde{B_2} + \widetilde{B_3} = (4S_0 + 4S_1 + 2S_2 + S_3) \bmod 8}  
Applying $\bmod \, 2$ to Eq. (\ref{eq:app2}), we get: 

\EQ{\widetilde{B_3} = (4S_0 \bmod 2 + 4S_1 \bmod 2 + 2S_2 \bmod 2 + S_3 \bmod 2) \bmod 2\nonumber
}
\begin{equation}\label{eq:app3}
\boxed{\widetilde{B_3} = S_3 \bmod 2 = \langle \widetilde{x}, {\alpha}^{(3)} \rangle \bmod 2}
\end{equation}
If, instead we apply $\bmod \, 4$ to Eq. (\ref{eq:app2}), we get: 

\EQ{2\widetilde{B_2} + \widetilde{B_3} &=& [4S_0 \bmod 4 + 4S_1 \bmod 4 + 2S_2 \bmod 4 + S_3 \bmod 4] \bmod 4\nonumber\\
2\widetilde{B_2} + \widetilde{B_3} &=& [(2S_2) \bmod 4 + S_3 \bmod 4] \bmod 4.\textrm{ Using \ref{I4}, we have:}\nonumber \\
2\widetilde{B_2} + \widetilde{B_3} &=& [2(S_2 \bmod 2) + S_3 \bmod 4] \bmod 4\nonumber
}
\EQ{\widetilde{B_2} &=& \frac{1}{2}\{[2(S_2 \bmod 2) + S_3 \bmod 4] \bmod 4 - S_3 \bmod 2\}. \textrm{ Using \ref{I6}:}\nonumber\\
\widetilde{B_2} &=& \frac{1}{2}[2(S_2 \bmod 2) + S_3 \bmod 4 - S_3 \bmod 2] \bmod 4\nonumber \\
\widetilde{B_2} &=& \frac{1}{2}\left\{2 \cdot \left[(S_2 \bmod 2) + \frac{(S_3 \bmod 4 - S_3 \bmod 2)}{2}\right] \right\} \bmod 4. \textrm{ Using \ref{I4}, we obtain:}\nonumber \\
\widetilde{B_2} &=& \left[S_2 \bmod 2 + \frac{(S_3 \bmod 4 - S_3 \bmod 2)}{2}\right] \bmod 2. \nonumber\\
\widetilde{B_2} &=& S_2 \bmod 2 \oplus \left( \frac{S_3 \bmod 4 - S_3 \bmod 2}{2} \right)\nonumber
}
\begin{equation}\label{eq:app4}
\boxed{\widetilde{B_2} = \langle \widetilde{x}, {\alpha}^{(2)} \rangle \bmod 2 \oplus \left( \frac{S_3 \bmod 4 - S_3 \bmod 2}{2} \right)}
\end{equation}
Finally, we can derive $\widetilde{B_1}$: 
\EQ{& \widetilde{B_1} = \frac{1}{4}\left\{(4S_0 + 4S_1 + 2S_2 + S_3) \bmod 8 - (S_3 \bmod 2) - \right.  \nonumber\\ 
& \left. 2\left[\left(S_2 \bmod 2 + \frac{(S_3 \bmod 4 - S_3 \bmod 2)}{2}\right) \bmod 2\right]\right\}\nonumber} 
Using \ref{I7}:  
\EQ{& \widetilde{B_1} = \frac{1}{4}\left\{(4S_0 + 4S_1 + 2S_2 + S_3 - (S_3 \bmod 2)) \bmod 8 - \right. \nonumber\\
& \left. 2\left[\left(S_2 \bmod 2 + \frac{(S_3 \bmod 4 - S_3 \bmod 2)}{2}\right) \bmod 2\right]\right\}\nonumber
} 
Using \ref{I5}: 
\EQ{& \widetilde{B_1} = \frac{1}{4}\left\{ 2 \left[\left(2S_0 + 2S_1 + S_2 + \frac{S_3 - S_3 \bmod 2}{2}\right) \bmod  4\right] -  \right. \nonumber\\
& \left. 2\left[\left(S_2 \bmod 2 + \frac{(S_3 \bmod 4 - S_3 \bmod 2)}{2}\right) \bmod 2\right]\right\}\nonumber\\
& \widetilde{B_1} = \frac{1}{2} \left[\left(2S_0 + 2S_1 + S_2 + \frac{S_3 - S_3 \bmod 2}{2}\right) \bmod 4 - \left(S_2 \bmod 2 + \frac{(S_3 \bmod 4 - S_3 \bmod 2)}{2}\right) \bmod 2\right]\nonumber} 
Using \ref{I6} we can rewrite the first term, and we get:
\EQ{\widetilde{B_1} &=& \frac{1}{2} \left\{ \left(S_2 + \frac{S_3 - S_3 \bmod 2}{2}\right) \bmod 2 + \left[ 2(S_0 + S_1) + S_2 + \frac{S_3 - S_3 \bmod 2}{2} - \right. \right. \nonumber\\
& & \left. \left. - \left(S_2 + \frac{S_3 - S_3 \bmod 2}{2}\right) \bmod 2 \right] \bmod 4 - \left[S_2 \bmod 2 + \frac{(S_3 \bmod 4 - S_3 \bmod 2)}{2}\right] \bmod 2 \right\}\nonumber
}
Combining the first and third term:
\EQ{\widetilde{B_1} &=& \frac{1}{2} \left\{ \left[ (S_2 - S_2 \bmod 2) + \frac{S_3 - S_3 \bmod 4}{2}\right] \bmod 2 \, + \right. \nonumber\\
& &\left. + \left[ 2(S_0 + S_1) + S_2 + \frac{S_3 - S_3 \bmod 2}{2} - \left(S_2 + \frac{S_3 - S_3 \bmod 2}{2}\right) \bmod 2 \right] \bmod 4 \right\} \nonumber
}
We notice that both $S_2 - S_2 \bmod 2$ and $\frac{S_3 - S_3 \bmod 4}{2}$ are even, so the first big term is $0$: 
\EQ{ \widetilde{B_1} = \frac{1}{2} \left\{ \left[ 2(S_0 + S_1) + S_2 + \frac{S_3 - S_3 \bmod 2}{2} - \left(S_2 + \frac{S_3 - S_3 \bmod 2}{2}\right) \bmod 2 \right] \bmod 4 \right\}\nonumber
} which can be rewritten as: 
\EQ{\widetilde{B_1} = \frac{1}{2} \left\{ \left\{ 2 \cdot \left[ S_0 + S_1 + \frac{\left(S_2 + \frac{S_3 - S_3 \bmod 2}{2}\right) - \left(S_2 + \frac{S_3 - S_3 \bmod 2}{2}\right) \bmod 2}{2}  \right] \right\} \bmod 4 \right\}\nonumber
}
Finally using \ref{I4}, we get:
\EQ{\widetilde{B_1} &=& \left[ S_0 + S_1  + \frac{\left(S_2 + \frac{S_3 - S_3 \bmod 2}{2}\right) - \left(S_2 + \frac{S_3 - S_3 \bmod 2}{2}\right) \bmod 2}{2} \right] \bmod 2\nonumber\\
\widetilde{B_1} &=& S_1 \bmod 2 \oplus S_0 \bmod 2 \oplus \left[ \frac{\left(S_2 + \frac{S_3 - S_3 \bmod 2}{2}\right) - \left(S_2 + \frac{S_3 - S_3 \bmod 2}{2}\right) \bmod 2}{2} \right] \bmod 2 \nonumber
}
\EQ{ \label{eq:app5} & \widetilde{B_1} = \langle \widetilde{x}, {\alpha}^{(1)} \rangle \bmod 2 \oplus \langle \widetilde{x}, b \rangle \bmod 2 \oplus \nonumber \\
& \left[ \frac{\left(S_2 + \frac{S_3 - S_3 \bmod 2}{2}\right) - \left(S_2 + \frac{S_3 - S_3 \bmod 2}{2}\right) \bmod 2}{2} \right] \bmod 2 }
\textbf{Important observation:} $\widetilde{B_1}$, $\widetilde{B_2}$, $\widetilde{B_3}$ all depend on the same value of $x$ and $x'$ (or $\widetilde{x}$, or $z$). Therefore, to make our analysis easier, we can consider that $z$ and $\widetilde{x}$ are fixed. Then, if we define the function: 
\EQ{B(r) = \langle \widetilde{x}, r \rangle \bmod 2 = \left( \sum \limits_{i = 1}^n \widetilde{x_i} r_i \right) \bmod 2
}
we can rewrite $\widetilde{B_3}, \widetilde{B_2}, \widetilde{B_1}$ as in Eq. (\ref{eq:proof2}) completing Step 1: 

\EQ{\widetilde{B_3} &=& B\left(\alpha^{(3)}\right)\nonumber\\
\widetilde{B_2} &=& B\left(\alpha^{(2)}\right) \oplus h_2(z, \alpha^{(3)})\nonumber \\
\widetilde{B_1} &=& B\left(\alpha^{(1)}\right) \oplus h_1(z, \alpha^{(3)}, \alpha^{(2)}, b)
}
Where: 
\EQ{\label{eq:app6}
& h_2(z, \alpha^{(3)}) := \frac{\langle z, {\alpha}^{(3)} \rangle \bmod 4 - \langle z, {\alpha}^{(3)} \rangle \bmod 2}{2} \\ 
& h_1(z, \alpha^{(3)}, \alpha^{(2)}, b) := \langle z, b \rangle \bmod 2 \, \,  \oplus  \\
 &\oplus \, \left[ \frac{ \left( \langle z, {\alpha}^{(2)} \rangle + \frac{\langle z, {\alpha}^{(3)} \rangle - \langle z, {\alpha}^{(3)} \rangle \bmod 2}{2} \right) - \left( \langle z, {\alpha}^{(2)} \rangle + \frac{\langle z, {\alpha}^{(3)} \rangle - \langle z, {\alpha}^{(3)} \rangle \bmod 2}{2} \right) \bmod 2 }{2} \right] \bmod 2 \nonumber}

\noindent\textbf{Step 2:} We see from Eq. (\ref{eq:app5}) that each of
the three bits involve a term similar to that of the GL theorem
\ref{thm:GL} (the $B\left(\alpha^{(i)}\right)$ term), but with two the
important differences. First, there is another term, and the bits of
$\widetilde{B}$ are XORs of the GL-looking term and that other
one. The second type of terms (that involve $h_1,h_2$) depend on
variables that appear in the expressions of other bits, potentially
introducing correlations among the different bits. We will deal with
the issue of correlations in Step 3, while with the effects of having
extra terms in Steps 4 and 5. Here we deal with the second important
difference, namely that the GL-looking terms (those of the form
$\langle \tilde{x},r\rangle\bmod2$) depend on $\tilde{x}$ rather than
$x$ in the inner product. For the remaining Step 2, we assume that the
first issue is resolved and it all reduces to GL theorem subject to
having $\tilde{x}$ rather than $x$.

Since we have $\tilde{x}$ in our expression if we follow the same
proof with that of the GL theorem we can follow the proof until the
point that we end up with obtaining a polynomial number of guesses for
$\tilde{x}$ of which one is the correct value with probability
negligibly close to unity. Now to continue with the proof we are
lacking two elements. First, in GL theorem the use the fact that
computing $f(x)$ given $x$ is easy, and check one-by-one the
polynomial guesses to see which one (if any) is correct. We cannot do
this since we only obtain $\tilde{x}$ and there is no way with no
extra information to check if $\tilde{x}$ actually corresponds to a
given image $y=f(x)=f(x')$. The second issue, is that even if we could
check this, having obtained $\tilde{x}$ does not contradicts the
definition of one-way function (definition
\ref{def:one_way_function}).

We resolve both these issues with two
observations. \textbf{Observation 1:} We notice that because of the
2-regularity property of $f$, $\widetilde{x}$ is uniquely determined
by $x$ ($f(x) = f(x')$, $\widetilde{x} = x \oplus
x'$). \textbf{Observation 2:} The assumption that our 2-regular
trapdoor function $f$ is second preimage resistant (i.e. a QPT
adversary given $x$, cannot find the second preimage $x'$, where
$f(x) = f(x'))$ means that:

\EQ{\underset{x \leftarrow \{0,1\}^n} {\Pr} [\mathcal{A}(1^n, x) = x' \text{ such that } f(x) = f(x')] \leq \negl \nonumber}
As 
\EQ{\underset{x \leftarrow \{0,1\}^n} {\Pr} [\mathcal{A}(1^n, x) = x'] = \underset{x \leftarrow \{0,1\}^n} {\Pr} [\mathcal{A}(1^n, x) = x' \oplus x]\nonumber
} 
we have that: \\
\EQ{\label{eq:app7}\underset{x \leftarrow \{0,1\}^n} {\Pr} [\mathcal{A}(1^n, x) = \widetilde{x}] \leq \negl
}

As we have mentioned, following the GL theorem proof we would obtain
polynomially many guesses for $\tilde{x}_g$ (where subscript $g$
stands for guess). Now by the second preimage resistance, if we are
given $x$ we should be unable to obtain $x'$ in polynomial
time. However, using our polynomially many guesses for $\tilde{x}$ and
checking for each guess if $f(x\oplus\tilde{x}_g)=f(x)$ we can obtain
with probability negligible close to unity the correct $\tilde{x}$ and
therefore come to contradiction with Eq. (\ref{eq:app7}).\\

\noindent\textbf{Step 3:} Since the different bits involve common
variables, to prove that our function is hard-core we need to consider
the issue of correlations. One way to deal with this would be to prove
the independence of both the bits and of the optimal guessing
algorithms. We, instead, use the \textit{Vazirani-Vazirani Theorem}
\ref{thm:VV}, which for our case it means that it suffices to show
that: $\widetilde{B_1}$, $\widetilde{B_2}$, $\widetilde{B_3}$,
$\widetilde{B_1} \oplus \widetilde{B_2}$,
$\widetilde{B_1} \oplus \widetilde{B_3}$,
$\widetilde{B_2} \oplus \widetilde{B_3}$,
$\widetilde{B_1} \oplus \widetilde{B_2} \oplus \widetilde{B_3}$ are
all hard-core predicates for $f$.

The most general expression that captures all these (to be proven)
hard-core predicates (formed from the subsets of
$\{\widetilde{B_1}, \widetilde{B_2}, \widetilde{B_1} \}$ ) is:

\EQ{ E(x, r_1, r_2, r_3) = \langle \widetilde{x}, r_1 \oplus r_2\rangle \bmod 2 \oplus g(z, r_2, r_3)}
where $g$ can be any binary function. Using $\langle \widetilde{x}, r_1 \oplus r_2\rangle \bmod 2=\langle \widetilde{x}, r_1\rangle \bmod 2\oplus \langle \widetilde{x}, r_2\rangle \bmod 2$ we can rewrite this as

\EQ{\label{eq:app8} E(x, r_1, r_2, r_3) = \langle \widetilde{x}, r_1\rangle \bmod 2 \oplus g'(z, r_2, r_3)}
where $g'(z,r_2,r_3)=\langle \widetilde{x}, r_2\rangle \bmod 2\oplus g(z, r_2, r_3)$. In other words, in order to prove that $\widetilde{B_1}\widetilde{B_2}\widetilde{B_3}$ is a hard-core function for $f$, it suffices to prove that $E(x, r_1, r_2, r_3)$ is a hard-core predicate for $f$.\\

\noindent\textbf{Step 4:} In this step, we will see how we can
effectively fix all but one variables, and turn Eq. (\ref{eq:app8}) to
depend only on $r_1$.

We want to prove that if there exists a QPT algorithm $\mathcal{A}$
that can guess the predicate $E$ as given in Eq. (\ref{eq:app8}),
$\mathcal{A} (f(x), r_1, r_2, r_3) = E(x, r_1, r_2, r_3)$ with
probability non-negligible better than $1/2$, then the second preimage
resistance assumption is violated by constructing a QPT algorithm
$\mathcal{A}'$ that, when given $x$ can obtain $\widetilde{x}$,
$\mathcal{A}'(f(x), 1^n) = \widetilde{x}$, with non-negligible
probability.

We now assume that the advantage $\mathcal{A}$ has in computing is
$\eps(n)$, without restricting $\eps(n)$ to be non-negligible,
aiming to reach a contradiction if this $\eps(n)$ is inverse
polynomial. We therefore assume:

\EQ{\label{eq:app9} \underset{\substack{ x \leftarrow \{0, 1\}^n \\ r_1 \leftarrow \{0, 1\}^{n} \\ r_2 \leftarrow \{0, 1\}^{n} \\ r_3 \leftarrow \{0, 1\}^{n}}} {\Pr} [\mathcal{A}(f(x), r_1, r_2, r_3) = E(x, r_1, r_2, r_3)] = \frac{1}{2} + \eps(n) 
}

Since the different variables ($x,r_1,r_2,r_3$) are chosen randomly
and independently we can effectively ``fix'' one variable.  We can
consider the set of values of that variable that satisfy some
condition that we need and name these values ``Good'' values (e.g. the
guessing algorithm $\mathcal{A}$ to succeed with higher than
negligible probability). Then we can work with the assumption that the
fixed variable is within the ``Good'' set, with only caveat that at
the end, whatever probability of inversion we obtain, is conditional
on the fixed variables being ``Good'' and thus we need to multiply
that probability with the probability that the fixed variable is
``Good''. For this reason, it is important that the probability of
being ``Good'' (ratio of cardinality of Good values with total values)
should be at least inverse polynomial.

We will, therefore, be using the following Lemma:
\begin{lemma}
  Let
  $\underset{(v_1,\cdots,v_k)| v_j\leftarrow\{0,1\}^n \forall
    j}{\Pr}[\text{Guessing}]\geq p+\eps(n)$, then for any variable
  $v_i$, $\exists$ a set $\text{Good}_{v_i}\subseteq\{0,1\}^n$ of size
  at least $\frac{\eps(n)}{2}2^n$, where $\forall$
  $\tilde{v}_i\in \text{Good}_{v_i}$

\[
\underset{(v_1,\cdots,\cancel{\tilde{v}_i},\cdots,v_k)| v_j\in\{0,1\}^n}{\Pr}[\text{Guessing}]\geq p+\frac{\eps(n)}2
\]
where the probability is taken over all variables except $v_i$.
\end{lemma}

\begin{proof}
\EQ{
& p+\eps(n) \leq \frac1{2^n}\sum_{v_i\in \text{Good}_{v_i}}\underset{(v_1,\cdots,\cancel{v_i},\cdots,v_k)| v_j\in\{0,1\}^n}{\Pr}[\text{Guessing}] + \nonumber\\  
& \frac1{2^n}\sum_{v_i\notin \text{Good}_{v_i}}\underset{(v_1,\cdots,\cancel{v_i},\cdots,v_k)| v_j\in\{0,1\}^n}{\Pr}[\text{Guessing}]\nonumber\\
& \leq \frac1{2^n}|\text{Good}_{v_i}|+\frac1{2^n}\sum_{v_i\notin \text{Good}_{v_i}}(p+\frac{\eps(n)}2)\nonumber\\}
\EQ{
p+\eps(n)&\leq& \frac1{2^n}|\text{Good}_{v_i}|+(p+\frac{\eps(n)}2)\nonumber\\
\frac{\eps(n)}{2}2^n&\leq&|\text{Good}_{v_i}|
}
\end{proof}
Now we return to Eq. (\ref{eq:app9}), we fix the set of 
$\text{Good}_x$, the set of inputs $x$ such that:
\EQ{\label{eq:app10}\underset{\substack{r_1 \leftarrow \{0, 1\}^{n} \\ r_2 \leftarrow \{0, 1\}^{n} \\ r_3 \leftarrow \{0, 1\}^{n}}} {\Pr} [\mathcal{A}(f(x), r_1, r_2, r_3) = E(x, r_1, r_2, r_3)] \geq \frac{1}{2} + \frac{\eps(n)}{2} \, \, \, \, \, \, \forall x \in \text{Good}_x}
and using \autoref{lemma:proof1} we have
$|\text{Good}_x|\geq \frac{\eps(n)}{2} 2^n$. Note that fixing $x$
is equivalent with fixing $\tilde{x}$ or $z$, given the definition of
the 2-regular function $f$. Starting with Eq. (\ref{eq:app10}) we can
now fix $r_3$ (conditional on $x\in\text{Good}_x$)

\EQ{\label{eq:app11}
\underset{\substack{r_1 \leftarrow \{0, 1\}^{n} \\ r_2 \leftarrow \{0, 1\}^{n}}} {\Pr} [\mathcal{A}(f(x), r_1, r_2, r_3) = E(x, r_1, r_2, r_3)] \geq \frac{1}{2} + \frac{\eps(n)}{4}  \nonumber\\ 
\, \, \,  \forall x \in \text{Good}_x \, \wedge \, \forall r_3 \in \text{Good}_{r_3}
}
where using again \autoref{lemma:proof1} we have $|\text{Good}_{r_3}|\geq \frac{\eps(n)}{4}2^n$. Finally, we can fix $r_2$ (conditional on  $x\in\text{Good}_x$ and $r_3\in\text{Good}_{r_3}$)

\EQ{\label{eq:app12}
\underset{\substack{r_1 \leftarrow \{0, 1\}^{n}}} {\Pr} [\mathcal{A}(f(x), r_1, r_2, r_3) = E(x, r_1, r_2, r_3)] \geq \frac{1}{2} + \frac{\eps(n)}{8} \nonumber\\
\, \, \, \, \, \, \forall x \in \text{Good}_x \, \wedge \, \forall r_3 \in \text{Good}_{r_3} \, \wedge \, \forall r_2 \in \text{Good}_{r_2} \label{eq:proba_good}
}
and again by \autoref{lemma:proof1} we have $|\text{Good}_{r_2}|\geq \frac{\eps(n)}{8}2^n$. \\ 

\noindent\textbf{Step 5:} In Eq. (\ref{eq:app12}) the only variable is
$r_1$. Using Eq. (\ref{eq:app8}) we can see that given that
$x,r_2,r_3$ are all fixed,
$E(x,r_1,r_2,r_3)=\langle\tilde{x},r_1\rangle\oplus g'(z,r_2,r_3)$
where $g'(z,r_2,r_3)=c$ is constant. Because $c$ is a constant, we can define $\tilde{\cA} = \cA \oplus c$. Now, we can easily see that:

\begin{align*}
  &\underset{\substack{r_1 \leftarrow \{0, 1\}^{n}}} {\Pr} [\cA(f(x), r_1, r_2, r_3) = \langle \widetilde{x}, r_1\rangle \bmod 2 \oplus g'(z, r_2, r_3)]\\
  =& \underset{\substack{r_1 \leftarrow \{0, 1\}^{n}}} {\Pr} [\tilde{\cA}(f(x), r_1, r_2, r_3) = \langle \widetilde{x}, r_1\rangle \bmod 2]
\end{align*}

% \EQ{\underset{\substack{r_1 \leftarrow \{0, 1\}^{n}}} {\Pr} [\mathcal{A}(f(x), r_1, r_2, r_3) = \langle \widetilde{x}, r_1\rangle \bmod 2 \oplus g'(z, r_2, r_3)]\leq \nonumber\\
% \underset{\substack{r_1 \leftarrow \{0, 1\}^{n}}} {\Pr} [\mathcal{A}(f(x), r_1, r_2, r_3) = \langle \widetilde{x}, r_1\rangle \bmod 2]\nonumber}
So, using Eq.~(\ref{eq:proba_good}), we obtain

\EQ{
\underset{\substack{r_1 \leftarrow \{0, 1\}^{n}}} {\Pr} [\tilde{\cA}(f(x), r_1, r_2, r_3) = \langle \widetilde{x}, r_1\rangle \bmod 2] \geq \frac{1}{2} + \frac{\eps(n)}{8} \nonumber\\
\, \, \, \, \, \, \forall x \in \text{Good}_x \, \wedge \, \forall r_3 \in \text{Good}_{r_3} \, \wedge \, \forall r_2 \in \text{Good}_{r_2} 
}
which is exactly the expression in GL theorem. There, one obtains
guesses for inversion, i.e. to obtain $\tilde{x}$ with a polynomial in
$\eps(n)$ probability of success, given the fixed
$x,r_2,r_3$'s. Multiplying this with the probability of actually being
in $\text{Good}_x$ and $\text{Good}_{r_3}$ and $\text{Good}_{r_2}$ we
obtain another polynomial in $\eps(n)$. This rules out the
possibility of $\eps(n)$ being inverse polynomial, since that
would break the second preimage resistance.  As we have already
stated, guessing $\tilde{x}$ with inverse polynomial success
probability does not contradict the one-way property of the trapdoor
function, but it does contradict the second preimage resistance, since
given $x$ and $\tilde{x}$ one can obtain deterministically $x'$.

Concretely, using GL proof to construct from $\tilde{\cA}$, a QPT
algorithm $\mathcal{A}'$ that obtains $\widetilde{x}$,
$\mathcal{A}'(f(x), 1^n) = \widetilde{x}$ for all inputs
$x \in \{0, 1\}^n$, when, $x \in \text{Good}_x$,
$r_3 \in \text{Good}_{r_3}$ and $r_2 \in \text{Good}_{r_2}$, this
algorithm $\mathcal{A'}$ succeeds with probability:

\EQ{\underset{\substack{x \leftarrow \{0, 1\}^{n} \\ r_3 \leftarrow \{0, 1\}^{n} \\ r_2 \leftarrow \{0, 1\}^{n}}} {\Pr} [\mathcal{A}'(f(x),1^n) = \widetilde{x}\, | \, x \in {\text{Good}_x} \, \wedge \, r_3 \in \text{Good}_{r_3} \, \wedge \, r_2 \in \text{Good}_{r_2}] \nonumber\\
 \geq \frac{\eps^2(n)}{2(32n + \eps^2(n))} \geq
 \frac{\eps^2(n)}{2(32n + 1)}\nonumber
}
Then, we have:

\EQ{& \underset{\substack{x \leftarrow \{0, 1\}^{n} \\ r_3 \leftarrow \{0, 1\}^{n} \\ r_2 \leftarrow \{0, 1\}^{n}}} {\Pr} [\mathcal{A}'(f(x), 1^n) = \widetilde{x}] \geq \nonumber\\
& \underset{\substack{x \leftarrow \{0, 1\}^{n} \\ r_3 \leftarrow \{0, 1\}^{n} \\ r_2 \leftarrow \{0, 1\}^{n}}} {\Pr} [\mathcal{A}'(f(x), 1^n) = \widetilde{x} \, \wedge \, x \in {\text{Good}_x} \, \wedge \, r_3 \in \text{Good}_{r_3} \, \wedge \, r_2 \in \text{Good}_{r_2}] \geq \nonumber \\
& \geq \underset{\substack{x \leftarrow \{0, 1\}^{n} \\ r_3 \leftarrow \{0, 1\}^{n} \\ r_2 \leftarrow \{0, 1\}^{n}}} {\Pr} [\mathcal{A}'(f(x), 1^n) = \widetilde{x} \, | \, x \in {\text{Good}_x} \, \wedge \, r_3 \in \text{Good}_{r_3} \, \wedge \, r_2 \in \text{Good}_{r_2}] \cdot \nonumber\\
& \cdot \underset{\substack{x \leftarrow \{0, 1\}^{n} \\ r_3 \leftarrow \{0, 1\}^{n} \\ r_2 \leftarrow \{0, 1\}^{n}}} {\Pr} [x \in {\text{Good}_x} \, \wedge \, r_3 \in \text{Good}_{r_3} \, \wedge \, r_2 \in \text{Good}_{r_2}]\nonumber
}
We now see that

\EQ{& \Pr[x \in {\text{Good}_x} \, \wedge \, r_3 \in \text{Good}_{r_3} \, \wedge \, r_2 \in \text{Good}_{r_2}] = \nonumber\\
& \Pr[r_2\in \text{Good}_{r_2}|r_3\in \text{Good}_{r_3}\wedge x\in \text{Good}_x]\cdot \Pr[r_3\in \text{Good}_{r_3}|x\in \text{Good}_x]\times \Pr[x\in \text{Good}_x]\nonumber
}
By construction we have $\Pr[x\in \text{Good}_x]=\frac{|\text{Good}_x|}{2^n}$ and $\Pr[r_3\in \text{Good}_{r_3}|x\in \text{Good}_x]=\frac{|\text{Good}_{r_3}|}{2^n}$ and $\Pr[r_2\in \text{Good}_{r_2}|r_3\in \text{Good}_{r_3}\wedge x\in \text{Good}_x]=\frac{|\text{Good}_{r_2}|}{2^n}$, which leads to

\EQ{\underset{\substack{x \leftarrow \{0, 1\}^{n} \\ r_3 \leftarrow \{0, 1\}^{n} \\ r_2 \leftarrow \{0, 1\}^{n}}} {\Pr} [\mathcal{A}'(f(x), 1^n) = \widetilde{x}]  &\geq& \frac{\eps^2(n)}{2(32n + 1)} \cdot \frac{|{\text{Good}_x}|}{2^n} \, \cdot \, \frac{|{\text{Good}_{r_3}}|}{2^n} \, \cdot \, \frac{|{\text{Good}_{r_2}}|}{2^n}  \nonumber\\
&\geq& \frac{\eps^5(n)}{128(32n + 1)}
}

where we can view $r_3$ and $r_2$ as the internal randomness of the
inversion algorithm $\mathcal{A}'$. It is clear that if $\eps(n)$
is non-negligible, it means that there exists polynomial $p(n)$ such
that $\eps(n)=1/p(n)$, and then

\EQ{\underset{\substack{x \leftarrow \{0, 1\}^{n} \\ r_3 \leftarrow \{0, 1\}^{n} \\ r_2 \leftarrow \{0, 1\}^{n}}} {\Pr} [\mathcal{A}'(f(x), 1^n) = \widetilde{x}]  &\geq& \frac{1}{128(32n + 1)p(n)^5}
}

which as explained in Step 2 breaks second preimage resistance,
Eq. (\ref{eq:app7}). Since all the terms given in Step 3
($\widetilde{B}_i,\widetilde{B}_i\oplus\widetilde{B}_j,
\widetilde{B}_1\oplus \widetilde{B}_2\oplus \widetilde{B}_3$) are of
the form $E(x,r_1,r_2,r_3)$ as in Eq. (\ref{eq:app8}) our analysis
suffices to prove that $\widetilde{B}_1\widetilde{B}_2\widetilde{B}_3$
is a hard-core function for $f$.

\end{proof}

\section{Proof of \autoref{Thm:bijective_to_desired}}\label{app:bijective}

\begin{lemma}[two-regular]
If $\mathcal{G}$ is a family of bijective functions, then $\mathcal{F}$ is a family of two-regular functions.
\end{lemma}

\begin{proof}

For every $y \in \Im f_{k'} \subseteq R$, where $k'=(k_1,k_2)$:
\begin{enumerate}
\item Since $\Im f_{k'}= \Im g_{k_1}$ and $g_{k_1}$ is bijective, there exist unique $x_1:=g^{-1}_{k_1}(y)$ such that $f_{k'}(x,0)=g_{k_1}(x)=y$.
\item Since $\Im f_{k'}= \Im g_{k_2}$ and $g_{k_2}$ is bijective, there exist unique $x_2:=g^{-1}_{k_2}(y)$ such that $f_{k'}(x,1)=g_{k_2}(x)=y$.\end{enumerate}
Therefore, we conclude that:

\EQ{\label{eq:preimages2}
\forall \ y \ \in \Im f_{k'}: \ f_{k'}^{-1}(y):=\{(g_{k_1}^{-1}(y),0),(g_{k_2}^{-1}(y),1)\} 
} 

\end{proof}

\begin{lemma}[trapdoor]
  If $\mathcal{G}$ is a family of bijective trapdoor functions, then $\mathcal{F}$ is a family of trapdoor functions.
\end{lemma}

\begin{proof}
Let $y \in \Im f_{k'} \subseteq R$. We construct the following inversion algorithm:\\

\procedure [linenumbering]{$Inv_{\mathcal{F}}(k', y, t_k')$} {
\pccomment{ $t_k' = (t_{k_1}, t_{k_2})$,  $k' = (k_1,k_2)$}\\
x_1 := Inv_{\mathcal{G}}(k_1, y, t_{k_1})\\
x_2 := Inv_{\mathcal{G}}(k_2, y, t_{k_2})\\
\pcreturn (x_1,0) \text{ and } (x_2, 1)
}

\end{proof}

\begin{lemma}[one-way]
If $\mathcal{G}$ is a family of bijective, one-way functions, then $\mathcal{F}$ is a family of one-way functions.
\end{lemma}

\begin{proof}
We prove it by contradiction. We assume that a PPT adversary $\mathcal{A}$ can invert any function in $\mathcal{F}$ with non-negligible probability $P$ (i.e. given $y\in \Im f_{k'}$ to return a correct preimage of the form $(x',b)$ with probability $P$). We then construct a PPT adversary $\mathcal{A'}$ that inverts any function in $\mathcal{G}$ with the same non-negligible probability $P$ reaching the contradiction since $\mathcal{G}$ is one-way by assumption.

From Eq. (\ref{eq:preimages2}) we know the two preimages of $y$ are: (i) $(g_{k_1}^{-1}(y),0)$ and  (ii) $(g^{-1}_{k_2}(y),1)$.  We now construct an adversary $\mathcal{A'}$ that for any function $g_k : D \rightarrow R$, inverts any output $y=g_k(x)$ with the probability $P/2$.\\

\procedure [linenumbering] {$\mathcal{A}'(k_c, y)$} 
{
r\sample \set{0,1}\\
k_r =k_c\\
(k_{r'}, t_{k_{r'}}) \sample Gen_{\mathcal{G}}(\secparam) \\
\pcif r=0 \pcthen \t k':=(k_r,k_{r'})\\
\pcelse\\
\t k':=(k_{r'},k_r)\\
\t (x', b) \gets \mathcal{A}(k', y) \\
\t \pcif ((b == r) \wedge (g_{k_r}(x') = y) \pcthen\\
\t \t \pccomment{$\mathcal{A}$ returns correct preimage that also corresponds to the challenge of $\mathcal{A'}$}	 \\
\t \t \pcreturn x' \\
\t \pcelse \pccomment{$\mathcal{A}$ failed in giving any of the preimages (happens with probability $1-P$)}\\
\t \t \pccomment{or the preimage returned corresponds to the $r'$ that is not the challenge (happens with probability $P/2$)}\\
\t \t \pcreturn 0 
}\\

The inversion algorithm succeeds with $1-((1-P)+P/2)=P/2$ and thus reaches a contradiction.
\end{proof}

\begin{lemma}[second preimage resistance] If $\mathcal{G}$ is a family of bijective,  one-way functions, then, any function $f \in \mathcal{F}$ is second preimage resistant.\end{lemma}

\begin{proof}
Assume there exists a PPT adversary $\mathcal{B}$ that given $k'=(k_1,k_2)$ and $(y,(x,b))$ such that $f_{k'}(x,b)=y$ can find $(x',b')$ such that $f_{k'}(x',b')=y$ with non-negligible probability $P$. From Eq. (\ref{eq:preimages2}) we know that the two preimages have different $b$'s. We now construct a PPT adversary $\mathcal{B'}$ that inverts the function $g_{k_c}$ with the same probability $P$, reaching a contradiction:\\

\procedure [linenumbering]{$\mathcal{B'}(k_c, y)$} {
(k_{2}, t_{k_{2}}) \sample Gen_{\mathcal{G}}(\secparam) \\
x_2=g_{k_2}^{-1}(y)\pccomment{ using the trapdoor $t_{k_2}$}\\
k' := (k_c, k_2) \\
(x, 0) \gets \mathcal{B}(k',y,(x_2,1)) \pccomment{where $y$ is an element from the image of $f_{k'}$} \\
\pcif f_{k'}(x,0)==f_{k'}(x_2,1)==y\\
\pcreturn x \\
\pcelse \pccomment{$\mathcal{B}$ failed to find a second preimage; happens with probability $(1-P)$}\\
\pcreturn 0
}

\end{proof}

\begin{lemma}[quantum-safe]
If $\mathcal{G}$ is a family of quantum-safe trapdoor functions, with properties as above, then $\mathcal{F}$ is also a family of quantum-safe trapdoor functions.
\end{lemma}

\begin{proof}
The properties that require to be quantum-safe is the one-wayness and second preimage resistance. Both these properties of $\mathcal{F}$ that we derived above were proved using reduction to the hardness (one-wayness) of $\mathcal{G}$. Therefore if $\mathcal{G}$ is quantum-safe, its one-wayness is also quantum-safe and thus both properties of $\mathcal{F}$ are also quantum-safe.
\end{proof}

\section{Proof of \autoref{thm:req_param}}\label{app:implementation}

In the following, we will denote by $f(s,e,c)$, the function {\tt REG2.Eval}$_\cP(k,(s,e,c))$ for $k$ the index function obtained by {\tt REG2.Gen}$_\cP(1^n)$, and by $s_0, e_0$ the trapdoor information associated with this function $f$.

We now prove separately the $\delta$-2 regularity, collision resistance, one-wayness and trapdoor property of the function in \autoref{def:REG2_fct}.

\subsection{$\delta$-2 regularity}

Here we describe how to achieve $\delta$-2 regularity using the construction {\tt FromInj} and specifically, the function in \autoref{def:REG2_fct}.

This reduces to ensuring that the two function inputs $(s, e)$ and $(s - s_0, e - e_0)$ both lie within the domain of the function. The input $(s, e)$ is the result of the inversion algorithm, so it is by definition inside the domain. Additionally, as the first element of the domain is only required to be in $Z_q^n$ and as $Z_q$ is closed with subtraction $\mod q$, then $s - s_0 \in Z_q^n$ for any $s$, $s_0$ $\in Z_q^n$.
On the other hand, the second element of the domain is required to be in $Z^m$, such that each component is bounded in absolute value by some value $\mu$. In this case, we are not guaranteed that adding or subtracting two such elements the result is still in the domain. What we want to ensure is that with (at least) constant probability over the choice of $(s, e)$ and $(s_0, e_0)$, the result $(s - s_0, e - e_0)$ is in the domain of the function.

It is not difficult to show that if $(s_0, e_0)$ is chosen arbitrarily from the domain of the function, then $(s - s_0, e - e_0)$ lies within the domain of the function only with inverse exponential in $m$ probability. 
This is why we consider restricting $e_0$ to be within a subset of the domain. By suitable choice of this subset we can make the success probability (of having two preimages) -- seen as a function in $m$ -- to be at least a constant value.

Firstly, we remark that the exact probability of success can be explicitly computed. Indeed, if the trapdoor noise $e_0$ is sampled from a Gaussian of dimension $m$, and standard deviation $\sigma$, and if the noise $e_1$ is sampled uniformly from an hypercube $C$ of length $2\mu$ (both distribution being centered on 0) then the probability that $e_0 + e_1$ is still inside $C$ is:
$$\left(\erf\left(\frac{\sqrt{2} \mu}{\sigma}\right) - \frac{\sigma}{\sqrt{2 \pi} \mu} \left(1-\exp\left(-2 \left(\frac{\mu}{\sigma}\right)^2\right)\right)\right)^m$$

However, for simplicity, and because we do not aim to find optimal parameters, we will use a (simpler) lower bound of this probability (that will be less efficient by a factor of $\sqrt{m}$). To do that, remark that using Lemma 2.5 in \cite{Regev}, we have that if $e_0 \in \mathbb{Z}^m$, such that each component of $e_0$ is sampled from a Gaussian distribution with parameter $\alpha' q$, then we have that every component of the vector $e_0$ is less than $\mu' := \alpha' q \sqrt{m}$ with overwhelming probability when $m$ increases. So one can remark that, up to a negligible term, the Gaussian distribution with parameter $\alpha' q$ is ``closer to 0'' than the uniform distribution on $[-\alpha' q \sqrt{m}; \alpha' q \sqrt{m}]$ for sufficiently large $m$ (i.e. for any $x$, the integral between $-x$ and $x$ of the Gaussian distribution is bigger, up to a negligible term, than the integral of the uniform distribution). Therefore, to obtain a lower bound on the probability of having two preimages, we can consider that $e_0$ is sampled according to the uniform distribution on a hypercube of length $2\alpha' q \sqrt{m}$ rather than according to the Gaussian distribution of parameter $\alpha' q$. This simplifies our analysis, and allows us to find the subset in which $e_0$ must reside, as seen in the following lemma. Note also that if you do not want to do any assumption on the input distribution, and only assume that the infinity norm is smaller than $\mu'$, then the same Lemma applies with the constant $4$ replaced by $2$.

\begin{lemma}[Domain Addition]\label{lemma:intersection_size}
  Let $V = \R^m$ be a vector space of dimension $m$, and let $\cD_{m,\mu}$ be the uniform distribution inside the hypercube of dimention $m$ and length $2\mu$ centered on 0. Then, for any $\mu' < \mu$, we have:
  \[P_{m, \mu,\mu'} := \Pr[||e_0 + e_1||_\infty \leq \mu \  | \  e_0 \leftarrow \cD_{m,\mu'}, e_1 \leftarrow \cD_{m,\mu}] = \left(1-\frac{\mu'}{4\mu}\right)^m\]
  Moreover, if $\mu' = O(\frac{\mu}{m})$ then the probability $P_{m, \mu,\mu'}$ becomes lower bounded by a positive constant.
\end{lemma}

\begin{proof}
 
As $ ||e_0 + e_1||_\infty $ must be less than $\mu$, which means that each component of the sum vector must be less than $\mu$, and as each component of the 2 vectors $e_0$ and $e_1$ was independently sampled, then we can simplify our proof by considering that $e_0$ and $e_1$ are vectors in $\R$, essentially determining $P_{1, \mu,\mu'}$ and then, we can compute $P_{m, \mu,\mu'} = {P_{1, \mu,\mu'}}^m$. \\
Then, let us denote by $E_1$ the random variable sampled uniformly from $[-\mu, \mu]$, $E_0$ the random variable sampled uniformly from $[-\mu', \mu']$ and $E$ the random variable obtained as $E = E_1 + E_0$. 
Therefore, ${P_{1, \mu,\mu'}} = Pr[-\mu \leq E \leq \mu]$. \\ 
Now, we can compute the density function of $E$ using convolution: 
$$f_E(e) = \int_{-\infty}^{\infty} f_{E_1}(e_1) \cdot f_{E_0}(e - e_1) \, de_1$$ where $f_{E_1}$ and $f_{E_0}$ are the probability density functions of $E_1$ and $E_0$ ($f_{E_1}(e_1) = \frac{1}{2\mu}$, when $e_1 \in [-\mu, \mu]$ and $0$ elsewhere and $f_{E_0}(e_0) = \frac{1}{2\mu'}$, when $e_0 \in [-\mu', \mu']$ and $0$ elsewhere). \\
Then, we are only interested in the cases when both the values of $f_{E_1}(e_1)$ and 
$f_{E_0}(e - e_1)$ are non-zero and for this we need to consider 3 cases for $e$, given by the intervals: $e \in [-\mu - \mu', \mu' - \mu] \, \bigcup \, [\mu' - \mu, \mu - \mu'] \, \bigcup \, [\mu - \mu', \mu + \mu'] $.
Thus, we can derive:
\begin{equation}
\begin{aligned}[b] 
	& f_{E}(e) = 
     \begin{cases}
      \int_{-\mu}^{e + \mu'} \frac{1}{4\mu\mu'} \, de_1 \text{ , } e \in [-\mu - \mu', \mu' - \mu]\\
      \int_{e - \mu'}^{e + \mu'} \frac{1}{4\mu\mu'} \, de_1 \text{ , } e \in [\mu' - \mu, \mu - \mu'] \\
      \int_{e - \mu'}^{\mu} \frac{1}{4\mu\mu'} \, de_1 \text{ , } e \in [\mu - \mu', \mu + \mu'] \\
     \end{cases} \nonumber
\end{aligned} 
\end{equation}

Finally, we have that $Pr[-\mu \leq E \leq \mu] = \int_{-\mu}^{\mu} f_{E}(e) \, de = 
\int_{-\mu}^{-\mu + \mu'} f_{E}(e) \, de + \int_{-\mu + \mu'}^{\mu - \mu'} f_{E}(e) \, de + \int_{\mu - \mu'}^{\mu} f_{E}(e) \, de = 
\int_{-\mu}^{-\mu + \mu'} \frac{e + \mu' + \mu}{4\mu\mu'} \, de + \int_{-\mu + \mu'}^{\mu - \mu'} \frac{1}{2\mu} \, de + \int_{\mu - \mu'}^{\mu} \frac{\mu + \mu' - e}{4\mu\mu'} \, de = 
1 - \frac{\mu'}{4\mu}
$. \\

Consequently, we have $P_{m, \mu,\mu'} = (1-\frac{\mu'}{4\mu})^m$.\\

  Now, given that $\mu$ is a function of $m$, $\mu = \mu(m)$, we want to determine the values of $\mu'$, such that this probability (seen as a function in $m$) is at least a positive constant number.
  \begin{itemize}
  \item If $\displaystyle \lim_{m \rightarrow \infty} \frac{\mu'}{\mu} = 0$, then:
    $$ \lim_{m \rightarrow \infty} \left(1 - \frac{\mu'}{4\mu}\right)^m = \lim_{m \rightarrow \infty} \left(1 - \frac{\mu'}{4\mu}\right)^{\frac{4\mu}{\mu'} \frac{\mu'm}{4\mu}} = \left(\frac{1}{e}\right)^{\lim_{m \rightarrow \infty} \frac{\mu'm}{4\mu}}$$
    Now, what we require is that $\displaystyle {\lim_{m \rightarrow \infty} \frac{\mu'm}{4\mu}}$ = c $\geq 0$, where $c$ is a constant, as then, we have that the probability of success is at least a constant $\geq \left(\frac{1}{e}\right)^c$. 

  \item If $ \displaystyle \lim_{m \rightarrow \infty} \frac{\mu'}{\mu} > 0$ (and less than 1, as $0 < \mu' < \mu$), then: 
    $$ \lim_{m \rightarrow \infty} \left(1 - \frac{\mu'}{4\mu}\right)^m = 0 $$
  \end{itemize}

  Consequently, it is clear that in order to get a positive constant lower bound for the success probability, we must have:
  $$ \mu'= c \cdot \frac{4\mu}{m}, \, \, \, c \geq 0 $$
\end{proof}

Thus, in our case, if $e_1$ is sampled uniformly on a hypercube of length $2\mu$ and $e_0$ from a Gaussian with parameter $\alpha' q$, by replacing the actual values of $\mu = \alpha q \sqrt{m}$ and $\mu' := \alpha' q \sqrt{m}$, what we require is that:
$$ \alpha' = c \cdot \frac{4\alpha}{m}, \, \, \, c \geq 0 $$

\subsection{Collision resistance}

We start by the observation that for the choices of \autoref{def:REG2_fct}, no PPT adversary can infer the trapdoor information $(s_0,e_0)$, as determining $s_0$ from $k = (A, b_0)$ would be equivalent to solving $\LWE{}_{q, \bar{\Psi}_{\alpha' q}}$:
\begin{corollary}[One-wayness of the trapdoor {\cite[Theorem 1.1]{Regev}} ]\label{thm:onewayness_secret}
Under the $\SIVP{}_\gamma$ (with $\gamma = \poly[n]$) assumption, no PPT adversary can recover the trapdoor information $(s_0, e_0)$.
\end{corollary}

\iffalse
\begin{lemma}[One-wayness of the secret]\label{thm:onewayness_secret}
  Assuming that there is no polynomial time algorithm that can solve $\SIVP{}_\gamma$ with $\gamma = \poly[n]$, and that the parameters are chosen so that \autoref{thm:req_param} applies, then with overwhelming probability it is not possible to recover the trapdoor part $(s_0, e_0)$ obtained when running the function {\tt REG2.Gen} with access to only the public key.
\end{lemma}

\begin{proof}
Because $e_0$ is sampled accordingly to a Gaussian of parameter $\alpha' q$ and $s_0$ is sampled uniformly over $\Z_p$, then recovering $s_0$ is exactly equivalent to solve $\LWE{}_{p, \bar{\Psi}_{\alpha' q}}$. But because $\alpha' q \geq 2 \sqrt{n}$ and $\alpha \in (0,1)$, then according to \cite[Theorem 1.1]{Regev}, it is possible to solve $\SIVP{}_\gamma$ for $\gamma = \tilde{O}(n/\alpha')$. But accordingly to our parameters, $n/\alpha' = \poly[n]$, so $\gamma = \poly[n]$. Then under the assumption that it is not possible to solve $\SIVP{}_\gamma$ for $\gamma = \poly[n]$, we know that with overwhelming probability it will not be possible to get the trapdoor $s_0$.
\end{proof}

\fi

\begin{lemma}[Collision resistance]\label{thm:collision_resistant}

  The function $f$ defined in \autoref{def:REG2_fct} is collision resistant if the parameters are chosen accordingly to \autoref{thm:req_param} assuming that $\SIVP{}_\gamma$ is hard.
\end{lemma}

\begin{proof}
  By contradiction, let us suppose that this function is not collision resistant. Then there exist two pairs $(s_1,e_1)$, $(s_2,e_2)$ such that $f(s_1,e_1,0) = y = f(s_2,e_2,1)$. Note that the last bits are necessary different since the two functions that fix the last bit, are injective when the error is smaller than $r_{max}$ (according to \cite[Theorem 5.4]{MP2012}). By the definition of $f$, $||e_1||_\infty \leq \mu$ and $||e_2||_\infty \leq \mu$, i.e both $e_1$ and $e_2$ has Euclidean norm smaller than $\sqrt{m} \mu$. Then, by definition, $y = f(s_2,e_2,1) = f(s_2,e_2,0) + f(s_0,e_0,0) = A(s_2+s_0)+(e_2+e_0)$. Now, we remark that with overwhelming probability (over the choice of the trapdoor), $||e_0||_2 \leq \mu' \sqrt{m}$ as stated in \cite[Lemma 2.5]{Regev}, so in this case, $||e_2+e_0||_2 \leq \sqrt{m}(\mu+\mu') \leq r_{max}$ (last assumption of \autoref{thm:req_param}). Then, according to \cite[Theorem 5.4]{MP2012}, there is exactly one element $(s,e)$ with $e$ of length smaller than $r_{max}$ such that $As + e = y$. Because $(s_1,e_1)$ is a solution, we then have that: $s_2+s_0 = s_1$ and $e_2+e_0 = e_1$, i.e. $e_0 = e_1-e_2$ and $s_0 = s_1 - s_2$. Hence, it is possible to deduce the trapdoor information $s_0$ and $e_0$ from the collision pair, which is impossible by \autoref{thm:onewayness_secret}.
\end{proof}

\subsection{One-wayness}
One could imagine that the one-wayness of the resulting function of \autoref{def:REG2_fct} is implied by the one-wayness of the function in \cite{MP2012} (as is the case in \autoref{lemma:onewayFromInj}). However, we need more care here, since in our construction the error term $e$ is not sampled from a Gaussian distribution with suitable parameters (unlike the error term $e_0$). \footnote{ While other ways to prove the one-wayness are possible, we give here one proof that uses the previous two lemmata.}

\begin{lemma}[Collision resistance to one-wayness] \label{thm:coll_resist_to_oneway}
  Let $f : A \rightarrow B \cup \bot$ (with $\bot \notin B)$, where $A$ is finite and can be efficiently sampled uniformly and let $C$ be the set of all $y \in B$ that admit 2 preimages. If the restriction of $f$ to the set $f^{-1}(B)$ is a collision resistant function that admits with non-negligible probability two preimages for any $y$ from its image and if $\frac{|f^{-1}(C)|}{|A|}$ is non-negligible, then $f$ restricted to the set $f^{-1}(C)$ is a one-way function.
\end{lemma}

\begin{proof}
By contradiction: suppose that $f$ is not one-way on $C$, i.e. with a non-negligible probability we can find a preimage of $y$ for $y$ uniformly sampled in $C$, and from this we can show how to find a collision. The idea is to sample an input $x \in A$, and then compute $y := f(x)$. Then, as $\frac{|f^{-1}(C)|}{|A|}$ is non-negligible, we know that with non-negligible probability this $y$ will have two preimages. Now, with non-negligible probability, this function will be easy to invert and one gets $x'$. Because we sample uniformly at the step before, we have the same probability to sample one image or the other, so with probability $1/2$, $x' \neq x$, therefore, we found a collision. 
%If any step failed, then we start again, and because all steps are non-negligible, in polynomial time we find a collision.
\end{proof}

\begin{corollary}[One-wayness from {\autoref{thm:collision_resistant}} and {\autoref{thm:coll_resist_to_oneway}} ]\label{thm:onewayness}
  The function defined in \autoref{def:REG2_fct} is one-way for all $y$ that admit two preimages, under the $\SIVP{}_\gamma$ hardness assumption, when the parameters are chosen accordingly to \autoref{thm:req_param}.
\end{corollary}
\iffalse
\begin{proof}
  This family of functions respects all the requirements needed in \autoref{thm:coll_resist_to_oneway}, and \autoref{thm:collision_resistant} tells us that the function is collision resistant, so this function is also one-way for all $y$ that admit two preimages.
\end{proof}
\fi

\subsection{Trapdoor}

We want to prove that using the trapdoor information of the {\tt REG2} construction, which consists of $(s_0,e_0)$ and $t_k$, the trapdoor information of the {\tt LWE} function, we can efficiently derive the preimages of an output $b$ of {\tt REG2.Eval}.
Firstly, we notice that to find all the preimages, we can simply run {\tt LWE.Inv} on $b$ as well as on $b-b_0$ and if we succeed we take only the preimages that lie in the input domain, i.e. whose error part $e$ is bounded in infinity norm by $\mu$: $||e||_\infty \leq \mu$. Because the function is injective, these are all the possible preimages.
However, because we are interested only in the case when there are exactly two preimages, the function {\tt REG2.Inv} can also do the following: we first run {\tt LWE.Inv} on $b$ and obtain $(s_1,e_1)$. Then, the inversion is completed by returning $(s_1,e_1,0)$ and $(s_1-s_0,e_1-e_0,1)$, which are both valid preimages, if and only if the function has two preimages (see \autoref{thm:collision_resistant} for more details).

\section{Proof of \autoref{thm:exist_param}}\label{app:implementation2}

\begin{proof}
  Using the following explicit values for the parameters of the Micciancio and Peikert injective trapdoor function \cite{MP2012}, we want to prove that they fulfil all of the requirements of \autoref{thm:req_param}:

  \begin{align*}
    n &= \lambda\\
    k &= 5\ceil{\log(n)} + 21\\
    q &= 2^k\\
    \bar{m} &= 2n\\
    \omega &= nk\\
    m &= \bar{m} + \omega\\
    \mu &= \ceil{2mn \sqrt{2+k}}\\
    \mu' &= \mu/m\\
    B &= 2
  \end{align*}
  and $\alpha, \alpha', C$ are defined as in \autoref{thm:req_param}. Now, let us proof that these parameters satisfy all the requirements.

  \begin{itemize}
  \item The first three requirements are trivially satisfied.
  \item In the forth condition, the only difficulty is to show that $\alpha < 1$. By definition,
    \begin{alignat*}{6}
      \alpha &= \frac{m \mu}{\sqrt{m} m q}
      & &=\frac{\mu}{\sqrt{m} q}
      & &= \frac{\ceil{2mn \sqrt{2+k}}}{\sqrt{m} q}
      & &\leq \frac{4mn \sqrt{2+k}}{\sqrt{m} q}
      & &\leq \frac{8mn \sqrt{k}}{\sqrt{m} q}\\
      &\leq \frac{8\sqrt{m}nk}{q}
      & &\leq \frac{8\sqrt{2n + nk}nk}{2^{21} n^5}
      & &\leq \frac{8\sqrt{2nk}nk}{2^{21} n^5}
      & &\leq \frac{16(nk)^{3/2}}{2^{21} n^5}\\
      &\leq \frac{16(n(5(\log(n)+1)+21))^{3/2}}{2^{21} n^5} \span\span\span
      &\leq \frac{16(5 \times 21 n^2)^{3/2}}{2^{21} n^5}
      & &\leq \frac{16 \times 1076 n^{3}}{2^{21} n^5}
      & &< \frac{1}{n^{2}} \leq 1
    \end{alignat*}
  \item Now, let us show the fifth condition, i.e. $\alpha' q \geq 2 \sqrt{n}$. First we note that $\alpha' q := \frac{\mu}{\sqrt{m}m} \geq 2 \sqrt{n} \Leftrightarrow \mu \geq 2 \sqrt{n}m\sqrt{m} = 2mn\sqrt{2+k}$. Then, by defining $\mu = \ceil{2mn\sqrt{2+k}}$, the condition is satisfied.
  \item For the fifth condition, i.e. $\frac{n}{\alpha'}$ is $\poly[n]$, we just need to remark that $1/\alpha' = \frac{m^{3/2}q}{\mu} < m^{3/2} q $, and that both $m$ and $q$ are $\poly[n]$.
  \item Finally, to show that the last condition is satisfied, we note that:
    \begin{align}
      \sqrt{m} \mu &< \frac{q}{2 B \sqrt{\left(C \cdot (\alpha \cdot q) \cdot (\sqrt{2n} + \sqrt{kn} + \sqrt{n})\right)^2+1}} - \mu'\sqrt{m}\\
                   &= \frac{q}{4 \sqrt{\left(C \cdot \frac{\mu}{\sqrt{m}} \cdot (\sqrt{2n} + \sqrt{kn} + \sqrt{n})\right)^2+1}}-\frac{\mu}{\sqrt{m}}
    \end{align}
    if and only if
    \[A := 4 \left(\sqrt{m} + \frac{1}{\sqrt{m}}\right) \mu \sqrt{\left(C \cdot \frac{\mu}{\sqrt{m}} \cdot (\sqrt{2n} + \sqrt{kn} + \sqrt{n})\right)^2+1} \leq q\]
    Now, let us suppose that $k := u \ceil{\log(n)} + v$ with $u \leq 5$ and $v \geq 19$ and we need to find $u,v$ such that $A \leq 2^k$. Note that we will include $v$ in some constants and then find the good $v$ at the end.
    First, remark that:
    \begin{align}
      \sqrt{m} + \frac{1}{\sqrt{m}}
      &= \sqrt{m}(1+\frac{1}{m})\\
      &= \sqrt{m}(1+\frac{1}{n(2+k)})\\
      &\leq \sqrt{m}(1+\frac{1}{2+k})\\
      &\leq \sqrt{m}\underbrace{(1+\frac{1}{2+v})}_{\gamma_0} = \gamma_0 \sqrt{m}
    \end{align}
    So now,
    \begin{align*}
      A &\leq 4 C \gamma_0 \mu^2 \sqrt{kn} \sqrt{\left(1 + \sqrt{\frac{2}{k}} + \frac{1}{\sqrt{k}} \right)^2+\frac{1}{kn\left(C \cdot \frac{\mu}{\sqrt{m}}\right)^2 } }\\
        &= 4 C \gamma_0 \ceil{2mn\sqrt{2+k}}^2 \sqrt{kn} \sqrt{\left(1 + \sqrt{\frac{2}{k}} + \frac{1}{\sqrt{k}} \right)^2+\frac{1}{kn\left(C \cdot \frac{\ceil{2mn\sqrt{2+k}}}{\sqrt{m}}\right)^2 } }\\
        &\leq 4 C \gamma_0 \ceil{2mn\sqrt{2+k}}^2 \sqrt{kn} \underbrace{\sqrt{\left(1 + \sqrt{\frac{2}{v}} + \frac{1}{\sqrt{v}} \right)^2+\frac{1}{v\left(2C \sqrt{2+v}\right)^2 } }}_{\gamma_1}\\
        &\leq 4 C \gamma_0 \gamma_1 \left(2mn\sqrt{2+k}+1\right)^2 \sqrt{kn}\\
        &= 4 C \gamma_0 \gamma_1 \left(2n^2(2+k)^{3/2}+1\right)^2 \sqrt{kn}\\
        &= 16 C \gamma_0 \gamma_1 n^4(2+k)^{3} \left(1+\frac{1}{2n^2(2+k)^{3/2}}\right)^2 \sqrt{kn}\\
        &\leq 16 C \gamma_0 \gamma_1 n^4(2+k)^{3} \underbrace{\left(1+\frac{1}{2(2+v)^{3/2}}\right)^2}_{\gamma_2} \sqrt{kn}\\
        &\leq 16 C \gamma_0 \gamma_1 \gamma_2 n^4\left(k\left(1+\frac{2}{k}\right)\right)^{3} \sqrt{kn}\\
        &\leq 16 C \gamma_0 \gamma_1 \gamma_2 n^4 k^3\underbrace{\left(1+\frac{2}{v}\right)^{3}}_{\gamma_3} \sqrt{kn}\\
        &\leq 16  C \gamma_0 \gamma_1 \gamma_2 \gamma_3 n^{9/2} \left(u\ceil{\log(n)} + v\right)^{7/2}\\
        &= 16 C \gamma_0 \gamma_1 \gamma_2 \gamma_3 n^{9/2} v^{7/2} \left(1 + \frac{ u \ceil{\log(n)}}{v}\right)^{7/2}\\  
        &\leq 16 C \gamma_0 \gamma_1 \gamma_2 \gamma_3 n^{9/2} v^{7/2} \left(1 + \frac{ 5 \ceil{\log(n)}}{19}\right)^{7/2}\\
        &\leq 16 C \gamma_0 \gamma_1 \gamma_2 \gamma_3 n^{9/2} v^{7/2} 3n^{1/2}\\
        &\leq 48 C \gamma_0 \gamma_1 \gamma_2 \gamma_3 v^{7/2} n^{5}
    \end{align*}
    Finally, we observe that if $v = 21$ and $u = 5$, we have $A \leq 2^{v+u\ceil{\log(n)}} = 2^k$, which concludes the proof.
\end{itemize}
\end{proof}

\end{appendices}
\newpage%
\bibliographystyle{alpha}
\bibliography{biblio}

\end{document}